\documentclass[sts]{imsart}

\usepackage{amsthm,amsmath, amssymb, color}
\RequirePackage[colorlinks,citecolor=blue,urlcolor=blue]{hyperref}

\usepackage{natbib}
\usepackage{color}
\usepackage{amsfonts, fancybox}
\usepackage{amssymb}
\usepackage[american]{babel}
\usepackage{graphicx, pstricks, pst-node, color}
\usepackage{psfrag,color}
\usepackage{dcolumn}
\usepackage{subfigure}
\usepackage{flafter, bm, dsfont}
\usepackage[section]{placeins}

\usepackage{multirow}
\usepackage{color}
\usepackage{amsmath}
\usepackage{float}
\usepackage{mathrsfs}
\usepackage{amsfonts}
\usepackage{dsfont}
\usepackage{amssymb}
\usepackage[boxed]{algorithm2e}
\usepackage{amssymb, latexsym}

\usepackage{graphicx}
\usepackage{epstopdf}

\usepackage{boxedminipage}

\newcommand{\bc}{\begin{center}}
\newcommand{\ec}{\end{center}}
\newcommand{\ba}{\begin{array}}
\newcommand{\ea}{\end{array}}
\newcommand{\be}{\begin{eqnarray}}
\newcommand{\ee}{\end{eqnarray}}
\newcommand{\bel}{\begin{eqnarray}\label}
\newcommand{\El}{\end{eqnarray}}
\newcommand{\bes}{\begin{eqnarray*}}
\newcommand{\Es}{\end{eqnarray*}}
\newcommand{\bn}{\begin{enumerate}}
\newcommand{\en}{\end{enumerate}}

\newcommand{\drho}{\partial_\rho}

\newcommand{\prok}{\p_{\hspace{-.6ex}\rho_k}}

\definecolor{MIT}{cmyk}{.24, 1.00, .78, .17} 
\definecolor{pink}{cmyk}{0, 1, 0, 0} 
\definecolor{darkgreen}{cmyk}{1,0, 1, 0} 

\newtheorem{lemma}{Lemma}
\newtheorem{definition}{Definition}
\newtheorem{proposition}{Proposition}

\newtheorem{theorem}{Theorem}

\newtheorem{remark}{Remark}

\newcommand{\abs}[1]{\left|#1\right|}

\newcommand{\cD}{\mathcal{D}}

\newcommand{\cF}{\mathcal{F}}

\newcommand{\cN}{\mathcal{N}}

\newcommand{\bX}{\mathbf{X}}

\newcommand{\R}{{\rm I}\kern-0.18em{\rm R}}
\newcommand{\h}{{\rm I}\kern-0.18em{\rm H}}
\newcommand{\K}{{\rm I}\kern-0.18em{\rm K}}
\newcommand{\p}{{\rm I}\kern-0.18em{\rm P}}
\newcommand{\E}{{\rm I}\kern-0.18em{\rm E}}
\newcommand{\Z}{{\rm Z}\kern-0.18em{\rm Z}}
\newcommand{\1}{{\rm 1}\kern-0.24em{\rm I}}
\newcommand{\N}{{\rm I}\kern-0.18em{\rm N}}

\newcommand{\pn}{\p_{\kern-0.25em n}}
\newcommand{\pnm}{\p_{\kern-0.25em n,m}}
\newcommand{\psubm}{\p_{\kern-0.25em m}}
\newcommand{\psubp}{\p_{\kern-0.25em p}}
\newcommand{\cfi}{\cF_{\kern-0.25em \infty}}

\newcommand{\argmin}{\mathop{\mathrm{argmin}}}

\newcommand{\ud}{\mathrm{d}}

\newcommand{\eps}{\varepsilon}

\newlength{\minipagewidth}
{
\begin{center}
\begin{boxedminipage}{0.8\linewidth}
\begin{center}
\textbf{\texttt{#1}}
\end{center}
\rm
\begin{tabbing}
....\=...\=...\=...\=...\=  \+ \kill
} %
{\end{tabbing}
\end{boxedminipage} \end{center} 
}
\newcolumntype{d}{D{.}{.}{6}}

\newcommand{\flo}[1]{\lfloor #1 \rfloor} 
\newcommand{\ceil}[1]{\lceil #1 \rceil} 
\newcommand{\mb}[1]{\mathbb{#1}} 
\newcommand{\mc}[1]{\mathcal{#1}} 
\newcommand{\rp}[1]{^{(#1)}} 
\newcommand{\uu}{\mathbf{u}\,} 
\newcommand{\y}{\mathbf{y}\,} 
\newcommand{\z}{\mathbf{z}\,} 

\newcommand{\mf}[1]{\mathbf{#1}} 
\newcommand{\sgn}{\text{sgn}} 
\newcommand{\mr}[1]{\mathrm{#1}} 
\newtheorem{claim}{Claim}[section]
\newtheorem{question}{Question}

\begin{document}

\begin{frontmatter}
		\title{Balancing Gaussian vectors in high dimension}
		\runtitle{Balancing Gaussian vectors}
		
		\author{ 
            \fnms{Paxton} \snm{Turner}\ead[label=paxton]{pax@mit.edu},
			\fnms{Raghu}
			\snm{Meka}\ead[label=raghu]{raghuvardhan@gmail.com}, \text{and}
			\fnms{Philippe} \snm{Rigollet}\thanksref{t2}\ead[label=rigollet]{rigollet@math.mit.edu}		
		}

 		\affiliation{UCLA and Massachusetts Institute of Technology  }
		\thankstext{t2}{Supported by NSF awards IIS-BIGDATA-1838071, DMS-1712596 and CCF-TRIPODS- 1740751; ONR grant N00014-17-1-2147.
		}

		\address{{Paxton Turner}\\
			{Department of Mathematics} \\
			{Massachusetts Institute of Technology}\\
			{77 Massachusetts Avenue,}\\
			{Cambridge, MA 02139-4307, USA}\\
			\printead{paxton}
		}
		
		\address{{Raghu Meka}\\
		    {Department of Computer Science} \\
			{UCLA and Massachusetts Institute of Technology}\\
			{ 3732H Boelter Hall } \\
			{ Los Angeles, CA 90095, USA } \\
			\printead{raghu}
		}
		
		\address{{Philippe Rigollet}\\
			{Department of Mathematics} \\
			{Massachusetts Institute of Technology}\\
			{77 Massachusetts Avenue,}\\
			{Cambridge, MA 02139-4307, USA}\\
			\printead{rigollet}
		}
		
 		\runauthor{Turner et al.}
		
		\begin{abstract}
		Motivated by problems in controlled experiments, we study the discrepancy of random matrices with continuous entries where the number of columns $n$ is much larger than the number of rows $m$. Our first result shows that if $\omega(1) = m = o(n)$,  a matrix with i.i.d. standard Gaussian entries has discrepancy $\Theta(\sqrt{n} \, 2^{-n/m})$ with high probability. This provides sharp guarantees for Gaussian discrepancy in a regime that had not been considered before in the existing literature. Our results also apply to a more general family of random matrices with continuous i.i.d entries, assuming that $m = O(n/\log{n})$. The proof is non-constructive and is an application of the second moment method. Our second result is algorithmic and applies to random matrices whose entries are i.i.d. and have a Lipschitz density. We present a randomized polynomial-time algorithm that achieves discrepancy $e^{-\Omega(\log^2(n)/m)}$ with high probability, provided that $m = O(\sqrt{\log{n}})$. In the one-dimensional case, this matches the best known algorithmic guarantees due to Karmarkar--Karp. For higher dimensions $2 \leq m = O(\sqrt{\log{n}})$, this establishes the first efficient algorithm achieving discrepancy smaller than $O( \sqrt{m} )$.
		\end{abstract}
		
		
	\end{frontmatter} 

    \section{Introduction}

Randomized controlled experiments are often dubbed the ``gold standard" for estimating treatment effects because of their ability to create a treatment and a control group that have the same features on average. Indeed,  pure randomization, i.e., assigning each observation uniformly at random between the treatment and control group, leads to two groups with approximately the same size, the same average age, the same average height, etc. Unfortunately, because of random fluctuations, this approach may not lead to the best balance between the attributes of the control group and those of the treatment group. Yet, near perfect balance is highly desirable since it often leads to a more accurate estimator of the treatment effect. This quest for balance was initiated at the dawn of controlled experiments. Indeed, W.S. Gosset, a.k.a Student (of $t$-test fame) already questioned the use of pure randomization when it leads to unbalanced covariates~\citep{Stu38}, and R.A. Fisher proposed randomized block designs as a better solution in certain cases~\citep{Fis35}. One traditional approach to overcome this limitation is to simply \emph{rerandomize} the allocation until the generated assignment is deemed balanced enough~\citep{MorRub12,LiDinRub18}. Rerandomization is effectively a primitive form of optimization that consists in keeping the best of several random solutions. However, it was not until recently that covariate balancing was recognized for the  combinatorial optimization problem that it really is. With this motivation,~\cite{BerJohKal15, Kal18} proposed algorithms based on mixed integer programming that, while flexible, did not come with theoretical guarantees. More recently,~\cite{HarSavSpiZha19} used new algorithms from \cite{BanDadGarLov18} with theoretical guarantees to generate experimental designs with a tunable degree of randomization versus covariate balance and characterized the resulting trade-off between model robustness and efficiency for a specific treatment effect estimator computed on data collected in such experiments.       

In this work, we investigate both the theoretical and algorithmic aspects associated to this question by framing it in the broader scope of \emph{vector balancing}. In particular, this question bears strong theoretical footing in discrepancy theory.\footnote{The recent work~\cite{HarSavSpiZha19} takes a similar point of view, though here our purpose is to focus purely on optimal covariate balance.} 
	
	Let $X_1, \ldots, X_n \in \mathbb{R}^m$ denote a collection of vectors and let $\bX$ denote the $m \times n$ matrix whose column $i$ is $X_i$. The \textit{discrepancy} $\mc{D}(X_1, \ldots, X_n)$ of this collection is defined as follows.\footnote{In the interest of clarity, we free ourselves from important considerations in the practical design of controlled experiments such as having two groups of exactly the same size.}
\begin{equation}
    \label{eq:defdisc}
	\mc{D}_n:=\mc{D}(X_1, \ldots, X_n) = \min_{\sigma \in \{ \pm 1 \}^n } \abs{\sum_{i = 1}^n \sigma_i X_i }_{\infty}= \min_{\sigma \in \{ \pm 1 \}^n } \abs{\bX\sigma}_{\infty} 
\end{equation}
	
	Discrepancy theory is a rich and well-studied area with applications to combinatorics, optimization, geometry, and statistics, among many others \citep[see the comprehensive texts][]{matousek,chazelle}. A fundamental result in the area due to \cite{spencer} states that if $\max_i |X_i|_\infty \leq 1$ and $m = n$, then $\mc{D}_n\le 6\sqrt{n}$.
	Spencer's proof is nonconstructive and relies on a technique known as \textit{partial coloring}. In the last decade, starting with the breakthrough work of~\cite{Ban10}, several algorithmic versions of the partial coloring method have been introduced to efficiently find a signing $\sigma$ that approximately attains the minimum in~\eqref{eq:defdisc}. These include approaches based on random walks \citep{Ban10,LovMek12}, random projections \citep{Rot17}, and multiplicative weights \citep{LevRamRot17}. In the regime where $m \geq n$, these algorithms can be used to compute a signing (or allocation) $\sigma \in \{-1,1\}^n$ with objective value $O(\sqrt{n \log(2m/n)} \,)$. Moreover, this guarantee is tight in the sense that examples are known with discrepancy matching this bound.
	

The aforementioned results make minimal structural assumptions on the vectors $X_1, \ldots, X_n$ and treat the input as worst-case. However, in the context of controlled experiments, it is natural to assume that $X_1, \ldots, X_n$ are, in fact, independent copies of a random vector $X \in \R^m$. While more general results are possible, the reader should keep in mind the canonical example where $X \sim \cN_m(0,I_m)$ is a standard Gaussian vector, and in particular where the entries of $X$ are of order 1. We dub the study of $\mc{D}_n$ in this context \emph{average-case discrepancy}. 

It was first shown in \cite{kklo} via a nonconstructive application of the second moment method that when $m = 1$, the average-case discrepancy is $\mc{D}_n=\Theta(\sqrt{n} \, 2^{-n})$ with high probability, assuming that the underlying distribution has a sufficiently regular density.  This result was extended to specific multidimensional regimes. First,~\cite{costello} showed that  $\cD_n=\Theta(\sqrt{n} \, 2^{-n/m})$ in the constant dimension regime $m=O(1)$. The optimal discrepancy is also known  in the super-linear regime $m \ge 2n$   where it was shown that $\cD_n= O(\sqrt{n \log(2m/n)})$.\footnote{The upper bound established in \cite{ChaVem14} presents additional polylogarithmic terms that are negligible for most of the range $m \ge 2n$. This is also the regime considered by~\cite{HarSavSpiZha19}.} In particular, there is a striking gap between this benchmark and the discrepancy $|\bX\sigma^{\mathsf{rdm}}|_\infty=\Theta(\sqrt{n \log m })$ achieved by a random signing $\sigma^{\mathsf{rdm}}$, especially in the sub-linear regime. Motivated by applications to controlled experiments, \cite{KriAzrKap19} studied the average-case discrepancy problem with the aim to improve on this gap. The authors devised a simple and efficient greedy scheme that, in the univariate case, outputs an allocation $\sigma^{\mathsf{gree}}$ satisfying $|\bX\sigma^{\mathsf{gree}}|=O(n^{-2})$. In addition,~\cite{KriAzrKap19} argue that  $|\bX\sigma^{\mathsf{gree}}|=O(n^{-2/m})$ for any \emph{constant} dimension $m$.

This state of the art leaves three important questions open: 
\begin{enumerate}
\item Can a sub-polynomial discrepancy be achieved in polynomial time even in dimension 1? 
\item What is the optimal discrepancy in the intermediate regime where $\omega(1) = m = o(n)$?
\item Do there exist efficient allocations that perform better than the random allocation in super-constant dimension? 
\end{enumerate}
The answer to the first question is well known. Indeed, the best known algorithm for number partitioning is due to \cite{KarKar82} and yields $\sigma \in \{-1,1\}^n$ such that  $|\bX\sigma|_\infty=e^{-\Omega(\log^2 n)}$ with high probability \citep[see also][]{BoeMer08}. While this result provides a super-polynomial improvement over algorithms built for the worst case, a significant gap remains between the information-theoretic bounds and the algorithmic ones despite extensive work on the subject \citep{BoeMer08,bcp,HobRamRotYan17}. This suggests the possibility of a statistical-to-computational gap similar to those that have been observed starting with sparse PCA~\citep{BerRig13,BerRig13b} and more recently in other planted problems~\citep{BreBreHul18,BanPerWei18}. Moreover, while the greedy algorithm of~\cite{KriAzrKap19} is loosely based on ideas from this algorithm, no multivariate extension of this algorithm is known even for the case $m=2$. Note that in the super-linear regime $m \ge 2n$, the work of~\cite{ChaVem14} also proposes a polynomial-time algorithm based on~\cite{LovMek12} showing an absence of substantial statistical-to-computational gaps. 

In this paper, we provide answers to the remaining two questions raised above. First, we
show that the discrepancy of standard Gaussian vectors is $\Theta(\sqrt{n} \, 2^{-n/m})$ with high probability for the remaining regime $\omega(1) = m = o(n)$. Moreover, we complement this existential result by giving the first randomized polynomial-time algorithm that achieves discrepancy $e^{-\Omega(\log^2(n)/m)}$ when $2\le m = O(\sqrt{ \log{n} })$. Note that while this remains an intrinsically low-dimensional result, it covers already super-constant dimension.  This first algorithmic result paves the way for potential algorithmic advances in a wider range of high-dimensional problems. In particular, our existential result sets an information-theoretic benchmark against which future algorithmic results can be compared as well as a baseline to establish potential statistical-to-computational gaps in high dimensions. These improved discrepancy bounds also have direct applications to randomized control trials. For example, in the case of an additive linear response with all covariates observed, the discrepancy attained by the allocation controls the fluctuations of the difference-in-means treatment effect estimator~\citep{KriAzrKap19}.  
	
Another point of view on balancing covariates in randomized trails is that of pairwise matching. In this setup, the experimenter first divides the sample into two equal-sized groups and then pairs up individuals who have similar covariates. For the unidimensional case, \cite{GreLuSil04} proposed a scheme that consists of performing a minimum cost matching that leads to a bounded discrepancy. This result may be extended to yield a discrepancy of order $n^{1-1/m}$ in dimension $m$ using results on random combinatorial optimization~\cite{Ste92}. Unlike matching algorithms, bipartite matching algorithms can be implemented in near-linear time using modern tools from computational optimal transport~\citep{CutPey18,AltWeeRig17,AltBacRud19}. We leave it as an interesting open question to study allocation schemes based on random bipartite matching problems for which sharp results have recently been discovered~\citep{LedZhu19}.



	\section{Main results}
In this section, we give an overview of our main results. Detailed computations and proofs are postponed to subsequent sections.

\subsection{Existential result}
	Our first main result shows that when $X_1, \ldots, X_n \stackrel{iid}{\sim} \mc{N}(0, I_m)$ and $\omega(1) = m = o(n)$, then the discrepancy is asymptotically $\sqrt{\frac{\pi n}{2}} \, 2^{-n/m} $ with high probability. As in the one-dimensional case \citep{kklo}, this result highlights that drastic cancellations are possible, with high probability, when the number of vectors grows asymptotically faster than the dimension. 
	
	\begin{theorem}
		\label{thm:sub-linear_discrepancy}
		Fix an absolute constant $\gamma>1$ and suppose that $\omega(1) = m = o(n).$ Let $X_1, \ldots, X_{n} \stackrel{iid}{\sim} \cN(0,I_m)$ be independent standard Gaussian random vectors. Then
		\begin{equation}	
		\label{eqn:upper_bound}
		\lim_{n \to \infty}\p\Big[\mc{D}(X_1, \ldots, X_n) \leq\gamma\sqrt{\frac{\pi n}{2}}2^{-n/m}\Big]=1 \,.
		\end{equation}
		If $\gamma' < 1$, then 
		\begin{equation}
		\label{eqn:lower_bound}
		\lim_{n \to \infty} \p\left[ \mc{D}(X_1, \ldots, X_n) \geq \gamma' \sqrt{\frac{\pi n}{2}} 2^{-n/m} \right] = 1.
		\end{equation}
	\end{theorem} 
	
	The work of \cite{costello} handles the case $m = O(1)$, and shows that the limiting probability in \eqref{eqn:upper_bound} is exactly $1 - \exp(-2 \gamma^m)$. We also note that the series of papers by \cite{bcp,bcmn1,bcmn2} provides an even more complete description of the unidimensional case.
	
	Our results are not limited specifically to Gaussian distributions. A mild extension of our techniques allows us to derive a similar result for a more general family of distributions, assuming that $m = O( n/\log n)$.

    \begin{remark}
    Let $C> 0$ denote a sufficiently small absolute constant, and suppose that $m \leq C n/\log n$. Let $\mathbf{X}$ denote an $m \times n$ random matrix whose entries are i.i.d random variables having a common density $f:\mb{R} \to \mb{R}$ such that 
    $$
    \int f(x)^2 \ud x <\infty, \qquad \int  x^4 f(x)\ud x<\infty, \quad \text{and} \quad  f(x)=f(-x), \forall\, x \in \R.\,
    $$
    Then there exist absolute positive constants $c\le c'$ such that
    \[
    \lim_{n \to \infty} \p\left[ c \sqrt{n} 2^{-n/m} \leq   \mc{D}(X_1, \ldots, X_n) \leq c'\sqrt{n} 2^{-n/m}  \right ] = 1. 
    \]
    \end{remark}
	We omit the proof of the above remark and focus on the Gaussian case for simplicity and because for Gaussian vectors, our analysis covers the whole range $m = o(n)$. 
	
	The proof of the upper bound in Theorem \ref{thm:sub-linear_discrepancy} is a nonconstructive application of the second moment method, in a similar spirit to the analysis of \cite{kklo} on the one-dimensional case as well as Achlioptas--Moore's analysis of the threshold for random $k$-SAT \citep{achlioptasmoore}. Recall that the second moment method states that for a nonnegative random variable $S$, we have
	\begin{equation}
	\label{eqn:second_moment}
	\p[S > 0] \geq \frac{\mb{E}[S]^2}{\mb{E}[S^2]}. 
	\end{equation} 
	As described in more detail in Section \ref{sec:sub-linear_discrepancy}, our strategy is to let $S$ count the number of signings with discrepancy at most $\gamma 2^{-n/m} \sqrt{\pi n/2} $ and show that the right-hand-side of \eqref{eqn:second_moment} tends to $1$ asymptotically. We also note that the lower bound in Theorem \ref{thm:sub-linear_discrepancy} is a straightforward consequence of the Markov inequality (first moment method) applied to $S$ (see Proposition \ref{prop:sub-linear_lbd}). 
	
	In addition to our result for $m = o(n)$, using similar techniques we also provide a precise characterization of Gaussian discrepancy in the linear regime $m \leq \delta n$, where $\delta$ is a sufficiently small absolute constant. In Appendix \ref{appendix:small_linear}, we show that the discrepancy is $\Theta(\sqrt{n}2^{-1/\delta})$ with probability at least $99\%$, asymptotically as $n \to \infty$. This provides further evidence of a conjecture of \cite{AubPerZde19} that the discrepancy when $m = \delta n$ is asymptotically $c(\delta) \sqrt{n}$ with high probability for an explicit function $c(\delta)$.\footnote{See Appendix \ref{appendix:small_linear} for a more precise description of their results.} In particular, our result combined with those of \cite{ChaVem14} confirms that the discrepancy is $\Theta(c(\delta) \sqrt{n})$ with asymptotic probability at least $99 \%$  when $m = \delta n$ for all $\delta > 0$. 
	
	Complementary to our work, we discuss recent existential results on average-case discrepancy in the discrete case when $X_1, \ldots, X_n$ are i.i.d vectors in $\{0, 1 \}^m$. Extending prior work of \cite{ezralovett}, \cite{frankssaks} and \cite{hobergrothvoss} use a nonconstructive Fourier-analytic argument to show, for two different models of random sparse binary vectors, that the discrepancy is $O(1)$ if $n= \tilde{\Omega}(m^3)$ \citep{frankssaks} and $n = \tilde{\Omega}(m^2)$ \citep{hobergrothvoss}. In addition, for the continuous case, \cite{frankssaks} show that the discrepancy of random unit vectors is $O( \exp(-\sqrt{n/m^3}))$. \cite{Pot18} uses the second moment method to show the discrepancy is $O(1)$ if $n = \Omega( m \log m )$ in the specific case where the entries of $X_1$ are uniform on $\{0, 1\}$. In other recent work, \cite{BanMek19} establish an average-case version of the Beck--Fiala conjecture, giving an algorithmic proof that the discrepancy of uniformly random $t$-sparse binary vectors is at most $O(\sqrt{t})$ for the entire range of parameters $m,n$ if $t = \Omega(\log \log m)$. It is an open question as to whether there exists a polynomial-time algorithm achieving $O(1)$ discrepancy for random $\{-1, +1\}$ vectors or sparse $\{0, 1\}$ vectors with $n = \mathsf{poly}(m)$ \citep{hobergrothvoss, frankssaks}.

\subsection{Algorithmic result}

	Our second main result is algorithmic and applies to a large family of continuous distributions. We construct a randomized polynomial-time algorithm called Generalized Karmarkar--Karp (\textbf{GKK}) that achieves discrepancy $\exp(-\Omega(\log^2(n)/m))$ with high probability, assuming $m = O(\sqrt{ \log{n}})$. This establishes the first such efficient algorithm achieving quasi-polynomially-small discrepancy for this regime. Our algorithm and analysis extend those of \cite{KarKar82} in the one-dimensional case to higher dimensions.\footnote{\cite{KarKar82} give two algorithms for number partitioning. The first one is a simple greedy heuristic, but its analysis was only performed for the uniform distribution over a decade later by \cite{Yak96}. Our algorithm presented here generalizes the second one which was rigorously analyzed in the original paper of \cite{KarKar82}.}
	
	\begin{theorem}
		\label{thm:gkk}
			Let $\mf{X}$ denote a random $m \times n$ matrix with iid entries having a common density $\rho: [-\Delta,\Delta] \to \mb{R}$ which is $L$-Lipschitz and bounded above by some constant $D > 0$. Suppose that 
			\[
			m \leq C\sqrt{\frac{\log n}{\max(1, \log \Delta)}}\,,
			\]for some sufficiently small absolute constant $C = C(D, L) > 0$. Then the algorithm \textbf{\emph{GKK}} outputs, in polynomial time, a signing $\sigma \in \{-1,+1\}^n$  such that
			\[ \abs{\mf{X}\sigma}_\infty \leq \exp \left( - \frac{c \log^2{n}}{m} \right)\,, \]
		with probability at least $1 - \exp(- c n^{1/4})$ for some absolute constant $c > 0$.
	\end{theorem} 
	
This result easily extends to distributions with unbounded support. For example, if $\mf{X}$ has i.i.d standard Gaussian entries, then setting $\Delta = O(\sqrt{\log{n}})$ and conditioning on the (high probability) event $\{ \abs{\mf{X}_{ij}} \leq \Delta \, \, \forall \, i, j\}$, we can apply Theorem \ref{thm:gkk} to show that \textbf{GKK} yields discrepancy $\exp( -c \log^2(n)/m)$ for the Gaussian matrix $\mf{X}$. 
	
	It is an open question as to whether or not the guarantee of Theorem \ref{thm:gkk} can be improved to achieve sub-quasi-polynomial discrepancy efficiently, even in dimension one. Note that for $m = 1$, \cite{HobRamRotYan17} provide evidence of hardness of a $O(2^{\sqrt{n}})$-approximation to the optimal discrepancy in worst case via a reduction from the Minkowski problem and the shortest vector problem. We leave the following question. 
	
	\begin{question}
		Suppose that $m = n^{\gamma}$ for some $\gamma \in (0,1)$. Let $\mf{X}$ denote a random $m \times n$ matrix with independent standard Gaussian entries. What is the smallest possible value of $|\mf{X} \sigma|_\infty$ that can be achieved algorithmically in polynomial time? 
	\end{question} In particular, it is an open problem as to whether the partial coloring method can be used to guarantee subconstant discrepancy for standard Gaussians when $m = n^\gamma$. We suspect that the answer is negative. It seems that even attaining discrepancy $o(\sqrt{m})$ serves as a natural bottleneck for such an approach. 
	

\section{Gaussian discrepancy in sub-linear dimension} 	
\label{sec:sub-linear_discrepancy}

The main goal of this section is to prove the following proposition. Throughout, we adopt the shorthand notation $u_n \lesssim_n v_n$ for $u_n \le v_n(1+o(1))$ and $u_n \simeq_n v_n$ for $u_n = v_n(1+o(1))$. 
	
	\begin{proposition}
		\label{prop:sub-linear}
		Fix $\gamma>1$, $\omega(1) = m = o(n),$ and let $X_1, \ldots, X_{n} \stackrel{iid}{\sim} \cN(0,I_m)$ be independent standard Gaussian random vectors. Then
		\[	\lim_{n \to \infty}\p\Big[\mc{D}(X_1, \ldots, X_n) \leq\gamma\sqrt{\frac{\pi n}{2}}2^{-n/m}\Big]=1 \,.\] 
		
	\end{proposition}

	We first outline our proof strategy based on the second moment method. Set $\eps = \eps(n) = \gamma 2^{-n/m}\sqrt{\pi n / 2}$  and define $S$, the number of low discrepancy solutions, to be
	\begin{equation}
	\label{eqn:defSvec}
	S = \sum_{\sigma \in \{\pm 1 \}^n} \1\big( \big|\sum_{i = 1}^{n} \sigma_i X_i\big|_\infty \leq \eps \big).
	\end{equation}
	Our goal is to show that $\E[S^2]/\E[S]^2=1+o(1)$. By the second moment method \eqref{eqn:second_moment}, this implies the desired result. 
	
	
	The next lemma gives a useful form for the first and second moments of $S$ and follows from a straightforward calculation. Its proof is postponed to Appendix \ref{appendix:sub-linear_discrepancy}.
	
	
	\begin{lemma}
		\label{lem:moments}
		The random variable $S$  defined as in \eqref{eqn:defSvec} has its first two moments given by
		\begin{equation}
		\label{eqn:1mm}
		\E[S] = 2^n \p\big( |Z| \leq \frac{\eps}{\sqrt{n}} \big)^m 
		\end{equation}
		where $Z \sim \cN(0,1)$, and
		\begin{equation}
		\label{eqn:2mm}
		\E[S^2] = 2^n  \sum_{k = 0}^n  \binom{n}{k}  \prok\left( \abs{\sqrt{n}X} \leq \eps \, , \, \abs{\sqrt{n}Y} \leq \eps  \right)^m\,.
		\end{equation}
	    Here $\rho_k = 1 - 2k/n$ and $\prok$ denotes the joint distribution of $(X,Y)$ with $X,Y\sim \cN(0,1)$ having correlation $\rho_k$. 
	\end{lemma}
	
\begin{figure}
	    \centering
			\includegraphics[scale = 0.5]{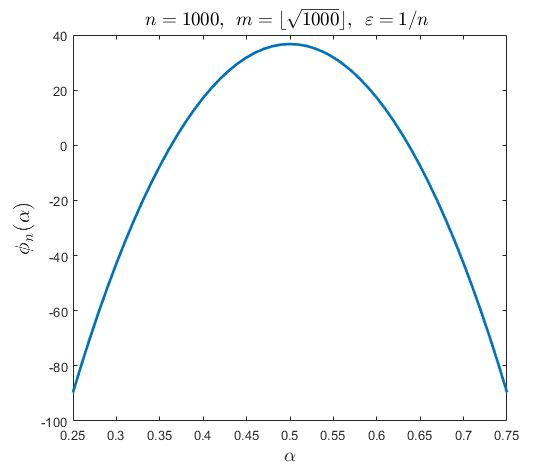}
			\caption{$\alpha \mapsto \phi_n(\alpha)$ for $n = 1000, m = \flo{\sqrt{1000}}$, and $\eps = 1/n.$}
			\label{FIG:log-concave1}
\end{figure}
	Given this representation, we proceed in two steps to prove an upper bound on the second moment $\E[S^2]$:
	%
	
	\begin{enumerate}
		\item[(i)] We first apply a truncation argument to show that the contribution from the $k \le n/4$ and $k \ge 3n/4$ terms in the summand of \eqref{eqn:2mm} is negligible. See Lemma \ref{lem:truncation} and its proof in Appendix \ref{appendix:sub-linear_discrepancy} for details. 
		\item[(ii)] Then we show that the dominant contribution in the summation \eqref{eqn:2mm} is asymptotically bounded by $\mathbb{E}[S]^2$ and comes from an interval of length $\Theta(\sqrt{n})$ around $k \simeq n/2$. This part is somewhat delicate and we apply the \emph{Laplace method} to obtain sharp bounds.
	\end{enumerate}

	
	
	
			
	By step (i), it suffices to control the leading term
	\begin{equation}
	\label{eqn:main_term}
	L:=2^n \sum_{k = n/4}^{3n/4} \binom{n}{k} \prok\left( \abs{\sqrt{n}X} \leq \eps \, , \, \, \abs{\sqrt{n}Y} \leq \eps  \right)^m.
	\end{equation}
	To that end, approximate the above binomial coefficient using~Lemma~C.2 in~\cite{BerRigSri18}: For any $l \in (0,1/2]$, $\alpha \in (l, 1-l)$ such that $n \alpha$ is an integer, it holds
	\[
	\exp\big(-\frac{1}{12 l^2 n}\big) \le  \sqrt{2\pi n \alpha(1-\alpha)}\exp(-n h(\alpha) ) \binom{n}{\alpha n} \le \exp\big(\frac{1}{12 n}\big) \,,
	\]
	where $h(\alpha)=-\alpha\log \alpha  -(1-\alpha)\log(1-\alpha)$ denotes the binary entropy with $h(0)=h(1)=0$. Therefore, it holds that
	\begin{equation}
	\label{eqn:Stirling}
	L \lesssim_n \frac{2^n}{\sqrt{2 \pi n}} \sum_{k = n/4}^{3n/4} \exp(\phi_n(\alpha_k))
	\end{equation}
	where $\alpha_k = k/n$ and 
	\begin{equation}
	\label{eqn:defphi_n}
	\phi_n(\alpha) = n h(\alpha) + m \log ( \p_{1 - 2\alpha}\left[ \abs{\sqrt{n}X} \leq \eps \, , \, \, \abs{\sqrt{n}Y} \leq \eps  \right]) - \frac{1}{2} \log \alpha(1-\alpha).
	\end{equation}

	Moreover, as justified in Lemma \ref{lem:logconcave} (see Appendix \ref{appendix:sub-linear_discrepancy}), for $n$ sufficiently large, $\phi_n(\alpha)$ is a strictly concave function on $[0.25, 0.75]$ with a unique maximum at $\alpha = 0.5$. See Figure \ref{FIG:log-concave1} for the graph of $\phi_n(\alpha)$ for a specific setting of the parameters. Thus we can make the Riemann sum approximation
		\begin{equation}
		\label{eqn:Riemann_sub-linear}
		L \lesssim_n  \frac{2^n}{\sqrt{2 \pi n}} \sum_{k = n/4}^{3n/4} \exp(\phi_n(\alpha_k)) \lesssim_n \frac{\sqrt{n} 2^n}{\sqrt{2 \pi }}   \int_{1/4}^{3/4} \exp(\phi_n(\alpha)) d\alpha.
		\end{equation} 
	
	Our goal now is to employ the Laplace method \citep[see, e.g.,][]{murray}, a well-known technique from asymptotic analysis, to compute explicitly the asymptotic growth of the right-hand-side above.	It consists in performing a second-order Taylor expansion of $\phi_n$ in order to reduce the problem to the computation of a Gaussian integral.
	
\begin{lemma} 
	\label{lem:Laplace}
	Suppose that $m=o(n)$ and set $\eps=\gamma2^{-n/m}\sqrt{n\pi/2}$. Recall the definition of $S$ from \eqref{eqn:defSvec}. Then
	\begin{equation}
	L \lesssim_n \mathbb{E}[S]^2. 
	\end{equation}
\end{lemma} 
	
	\begin{proof}
		 We apply the Laplace method to the integral in \eqref{eqn:Riemann_sub-linear}. Let $\eta \in (0,1)$ be arbitrary, and define $g_n(\alpha) = \phi_n(\alpha)/n$. Since $h''(\alpha)$ is continuous, Lemma \ref{lem:logconcave} implies that there exists $\delta = \delta(\eta)$ and $N = N(\eta)$ such that
		\begin{equation}
		\label{eqn:phi_bound}
		\frac{1}{n} \abs{ \phi_n''(\alpha) - \phi_n''(1/2) } \leq \eta\, , \quad \forall \alpha \in (1/2 - \delta, 1/2 + \delta), \, n \geq N.  
		\end{equation}
		The above inequality follows by writing $g_n''(\alpha) = h''(\alpha) + r_n(\alpha)$, where $r_n(\alpha)$ is a remainder term that goes to $0$ uniformly in $\alpha \in (0.25, 0.75)$ as $n \to \infty$, using Lemma \ref{lem:logconcave}. Using that the remainder term is small and $h''(\alpha)$ is continuous at $\alpha = 1/2$, we arrive at \eqref{eqn:phi_bound}. 
		
		By \eqref{eqn:phi_bound} and Taylor's theorem,
		\begin{equation}
		\label{eqn:Taylor}
		\phi_n(\alpha) - \phi_n(1/2) \leq \frac{1}{2} (\phi_n''(1/2) + \eta n)(\alpha - 1/2)^2 \, , \quad \forall \alpha \in (1/2 - \delta, 1/2 + \delta), \, n \geq N.
		\end{equation}
		Moreover,
		\begin{equation}
		\label{eqn:negative2derivative}
		\phi_n''(1/2) + \eta n < 0
		\end{equation}
		for $n$ sufficiently large because $\eta \in (0,1)$ and $\phi_n''(1/2) \simeq_n -4n$ by Lemma \ref{lem:logconcave}. Therefore, since $\phi_n$ is increasing on $(0.25, 0.75)$ for $n$ sufficiently large,
		\begin{align}
		\label{eqn:Laplace_truncation}
		\frac{\sqrt{n}}{\exp(\phi_n(1/2))}\int_{1/4}^{1/2 - \delta} \exp( \phi_n(\alpha) ) \, d\alpha 
		& \lesssim_n 10 \sqrt{n} \exp( \phi_n(1/2 - \delta) - \phi_n(1/2) ) \\
		& \lesssim_n 10 \sqrt{n} \exp\left( \frac{1}{2} (\phi_n''(1/2) + \eta n)\delta^2 \right) = o(1) \nonumber,
		\end{align} 
		where we applied \eqref{eqn:Taylor} and \eqref{eqn:negative2derivative}. By symmetry of $\phi_n(\alpha)$ about $\alpha = 1/2$, the integral as in \eqref{eqn:Laplace_truncation} from $1/2 + \delta$ to $3/4$ is negligible. Moreover, by \eqref{eqn:Taylor},
		\begin{align}
		\label{eqn:Laplace_main}
		\int_{1/2 - \delta}^{1/2 + \delta} \exp(\phi_n(\alpha)) \, d\alpha &\lesssim_n 
		\int_{1/2 - \delta}^{1/2 + \delta}  \exp\left(\phi_n(1/2) + \frac{1}{2}(\phi_n''(1/2) + \eta n)(\alpha - 1/2)^2 \right) \, d\alpha \\
		&\lesssim_n \exp(\phi_n(1/2)) \sqrt{ \frac{2 \pi }{\abs{\phi_n''(1/2) + \eta n}} } =  2^n f_n^m \sqrt{\frac{2 \pi}{n(1 - \eta/4)}}, \nonumber
		\end{align}
		where 
		\[
		f_n = \p_{0}( \abs{\sqrt{n} X} \leq \eps, \abs{\sqrt{n} Y} \leq \eps ).
		\]
		Since $\eta \in (0,1)$ was arbitrary, we conclude by \eqref{eqn:1mm}, \eqref{eqn:Stirling}, \eqref{eqn:Riemann_sub-linear}, \eqref{eqn:Laplace_truncation}, \eqref{eqn:Laplace_truncation}, and the definition of $f_n$ that
		\[
		L \lesssim_n \frac{2^n}{\sqrt{2 \pi n}} \cdot n \cdot \int_{1/4}^{3/4} \exp(\phi_n(\alpha)) \, d \alpha \lesssim_n 2^{2n} f_n^m = \mathbb{E}[S]^2. 
		\]
	\end{proof}
	
	
%

	\begin{proof}[Proof of Proposition \ref{prop:sub-linear}]
		We see that $\mathbb{E}[S^2]/\mathbb{E}[S]^2 \lesssim_n 1$ as $n\to \infty$ applying Lemma \ref{lem:moments}, Lemma \ref{lem:truncation}, \eqref{eqn:main_term}, \eqref{eqn:Stirling}, and Lemma \ref{lem:Laplace}. Proposition \ref{prop:sub-linear} follows by the second moment method.
	\end{proof}
	
	We complement Proposition \ref{prop:sub-linear} with a near-matching lower bound.
	
	\begin{proposition}
		\label{prop:sub-linear_lbd}
		Let $\omega(1) = m = o(n)$, fix $\gamma < 1$, and let $X_1, \ldots, X_n \stackrel{iid}{\sim} \mc{N}(0, I_m)$ be independent standard Gaussian random vectors. Then
		\[
		\lim_{n \to \infty} \p\left[ \mc{D}(X_1, \ldots, X_n) \leq \gamma \sqrt{\frac{\pi n}{2}} 2^{-n/m} \right] = 0.
		\]
	\end{proposition}
	
	\begin{proof}
		Recall the definition of $S$ as in \eqref{eqn:defSvec}, which counts the number of signings with discrepancy $\eps = \gamma  2^{-n/m} \sqrt{\pi n / 2}$. By the Markov inequality, \eqref{eqn:gauss_approx}, and \eqref{eqn:1mm},
		\[
		\p\left[ S > 1 \right] \leq \mathbb{E}[S] = 2^n \p\left[|Z| < \gamma \sqrt{\frac{\pi n}{2}} 2^{-n/m} \right]^m \lesssim_n \gamma^m \to 0
		\]
		because $\omega(1)= m = o(n)$ and $\gamma < 1$. This completes the proof.
	\end{proof}
	
	Our first main result, Theorem \ref{thm:sub-linear_discrepancy}, is a direct consequence of Propositions \ref{prop:sub-linear} and \ref{prop:sub-linear_lbd}.

	\section{ Algorithmic discrepancy minimization in low dimension } 
	\label{sec:algo_discrepancy}
	
	Now we describe our approach for proving Theorem \ref{thm:gkk}. In this section we introduce the generalized Karmarkar--Karp algorithm \textbf{GKK}. Recall that the goal is to find algorithmically $\sigma \in \{ \pm 1 \}^n$ such that $\abs{\mathbf{X} \sigma}_\infty$ is small. As in \cite{KarKar82}, our algorithm is a \textit{differencing method}, which means that throughout the algorithm, we maintain a set of vectors $S$, and our basic operations consist of removing two vectors, say $x$ and $y$, from $S$ and then adding the difference to $S$ : $S \leftarrow S \cup \{x - y\} \backslash \{x, y\}$. We perform a sequence of these differencing operations in a judicious way until there is a single vector $v$ remaining in $S$. Note that at any given time, the elements of $S$ correspond to (disjoint) partial signed sums of the original vectors $X_1, \ldots, X_n$. Hence, the final vector $v \in S$ is indeed a signed sum of the original vectors. It is possible to keep track of the final signing by tracking the differences, though we do not do so explicitly.
	
	Next, we informally describe the \textbf{GKK} differencing method in detail. For simplicity, we assume that $\Delta = 1$ in this description. The algorithm \textbf{GKK} is a recursive procedure that consists of $\Theta(\log{n})$ phases. For the first phase of the recursion, given a collection of $n$ vectors lying in $[-1,1]^m$, we partition this cube into sub-cubes of side length $\alpha = n^{-\Omega(1/m)}$. The idea is that with sub-cubes of this size, we are likely to have multiple points in each sub-cube, and these points would be very close to each other. We then randomly difference the vectors in each sub-cube until there is at most one point left in each sub-cube. Next, we enter a \emph{clean-up} step to deal with the leftover vectors. First we combine the leftover vectors (at most one per each sub-cube) via a standard differencing algorithm that we call \textbf{REDUCE} into a single `bad' vector $v\rp{0}$ and let $G' \subset [-\alpha, \alpha]^m$ denote the vectors formed from random differencing. Next we make the entries of the bad vector small by adding signed combinations of a few vectors from $G'$. Namely, we draw at random points from $G'$ and greedily difference them against $v\rp{0}$ until the resulting vector is sufficiently small in the Euclidean norm. Specially, our update procedure for this clean-up step is 
	\begin{align}
	    	    v\rp{k} &= v\rp{k-1} + a^* \uu_k 	\label{eqn:clean_up_update} \\
	    	    a^* &= \argmin_{a \in \{ \pm 1 \} } \abs{ v\rp{k-1} + a \uu_k}_2. \nonumber
	\end{align}
    where $\uu_k$ is drawn at random from the remaining vectors is $G'$.
	
	Once we have $v\rp{k} \in [-O_m(\alpha), O_m(\alpha)]^m$, we stop drawing random vectors from $G'$, and this ends the first phase of recursion. The remaining vectors form the input to the second phase, which applies the same procedure as above on the smaller cube $[-\alpha, \alpha]^m$. Moreover, subsequent phases follow the same pattern: \textbf{partition}, \textbf{difference}, and \textbf{clean-up}. After each phase, the input cube shrinks by a factor of $n^{-\Omega(1/m)}$. Hence, after a logarithmic number of phases, the remaining vectors lie in a cube of side length $n^{-\Omega(n/m)} = e^{-\Omega(\log^2n/m)}$. We then apply \textbf{REDUCE} to combine the remaining vectors into a single vector with discrepancy as in Theorem \ref{thm:gkk}. 
	
	We remark that our algorithm also features a resampling step that happens immediately after partitioning. In each phase, this resampling procedure labels points as `good' or `bad' so that the good points are independent and have independent coordinates that have a nice distribution. This same resampling trick was also used in \cite{KarKar82} and is essential for (most of) the remaining random vectors at the end of each phase to have a nice distribution facilitating a recursive analysis. Moreover, the \textbf{partition} and \textbf{difference} steps of our algorithm are also similar to those used in \cite{KarKar82} for the one-dimensional case. 
	
	
	In summary, the algorithm \textbf{GKK} consists of several phases of a subroutine \textbf{PRDC}, which stands for partition, resample, difference, clean-up, that we now define explicitly. In the first part of the clean-up phase, we remark that the aforementioned algorithm \textbf{REDUCE} is applied. However, we defer the explicit description of this algorithm, which uses standard techniques, to Appendix \ref{appendix:reduce}, instead stating its key property of use.
	
	\begin{lemma}
		\label{lem:reduce}
		Given $X_1, \ldots, X_N \in \mathbb{R}^m$, the algorithm \textbf{\emph{REDUCE}} is polynomial-time and outputs $\sigma \in \{\pm 1\}^N$ such that 
		\begin{equation} 
		\label{eq:reduce_upper_bound}
		\abs{\sum_{i = 1}^N \sigma_i X_i}_\infty \leq \max_{ S \subset [N]: |S| = m } \sum_{j \in S} \abs{X_j}_\infty. 
		\end{equation}
	\end{lemma}
	
	In the explicit description of \textbf{PRDC} below, $\gamma > 0$ denotes a fixed absolute constant to be set later (see Appendix \ref{appendix:efficient_clean_up}). 
	
	\medskip
	
	
	
	\vspace{0.3 cm}
	\noindent \textbf{PRDC}:
	
	\textbf{Input}: A number $\alpha_t > 0$. A set of vectors $S_t \subset [-\alpha_t, \alpha_t]^m$. A single vector $v_t \subset \gamma m[-\alpha_t, \alpha_t]^m$. A pdf $g_t:[-\alpha_t, \alpha_t]^m \to \mb{R}$. Define $N_{t} = 2^m \ceil{|S_t|^{1/(4m)}}^m$. 
	\begin{enumerate}
		
		\item \textbf{Partition:} Define $\alpha_{t+1} = \alpha_t/\ceil{|S_t|^{1/(4m)}}$. Divide the cube $[-\alpha_t, \alpha_t]^m$ into $N_{t}$ disjoint sub-cubes $C_1, \ldots, C_{N_t}$ that are of the form $\alpha_{t+1} z + [0, \alpha_{t+1}]^m$ for some integer vector $z \in \mb{Z}^m$.  
		\item \textbf{Resample:} Independently for every vector $x$ in $S_t$, if $x \in C_{j}$, then label $x$ as `good' with probability $(\min_{y \in C_j} g_t(y))/g_t(x)$. Otherwise, label $x$ to be `bad.' Let $G_t$ denote the set of good points and $B_t$ denote the set of bad points. 
		\item \textbf{Difference:} For every sub-cube $C_j$, pick uniformly at random two points in $G_t \cap C_j$, include their difference in $G'_t$, and remove them from $G_t$. Continue this until $G_t \cap C_j$ has at most $1$ good point for every $j$. Let $B_t'$ be the union of $B_t, v_t$, and the leftover good points.
		\item \textbf{Clean-up:} 
		\begin{enumerate}
			\item Apply \textbf{REDUCE} to the vectors in $B'_{t}$ to obtain $\sigma$. Define $v\rp{0}_t = \sum_{b_i \in B'_{t}} \sigma_i b_i$. 
			\item For $k = 0, 1, 2, \ldots$ 
			
			\indent If $\abs{v\rp{k}_t}_2 \geq \gamma m \alpha_{t+1}$: remove uniformly at random a point $x \in G_{t}'$. Define $v\rp{k+1}_t = v\rp{k}_t + a^* x$ where $a^* = \argmin_{a \in \{ \pm 1\}} |v\rp{k}_t + a x  |_2$. Define $G_{t}' \leftarrow G'_{t} \backslash \{ x \}$. 
			
			\indent Else: $v_{t+1} := v\rp{k}_t$. BREAK
		\end{enumerate}
	\end{enumerate}
	
	\textbf{Output:} $S_{t+1} := G'_{t}$, $v_{t+1}$, $\alpha_{t+1} := \alpha_t/\ceil{|S_t|^{1/(4m)}}$ 
	
	\vspace{0.3 cm}
	
	
	Now we explicitly describe our main algorithm \textbf{GKK} in terms of the subroutine \textbf{PRDC}. Recall that $\rho$ is the density corresponding to a particular entry of $\mf{X}$. First we need the following definition. 
	\begin{definition}[Triangular distribution]
		\label{def:tri_distr}
		A random vector $\y \in \mathbb{R}^m$ follows a \emph{triangular distribution} on the cube $[-R, R]^m$ if the distribution of $\y$ is given by $\mathbf{u} - \mathbf{v}$, where $\mathbf{u}$ and $\mathbf{v}$ are independent and uniformly distributed on $[0, R]^m$. Notationally, we write $\y \sim \mr{Tri}[-R, R]^m$. 
	\end{definition}
	
	\vspace{0.3 cm}
	\noindent \textbf{GKK:}
	
	\textbf{Input:} An $m \times n$ matrix $\mf{X}$. A probability density function $\rho:[-\Delta, \Delta] \to \mathbb{R}$. 
	Let $T = \ceil{ C^* \log n }$ where $C^*:= (2\log(10/3))^{-1}$.   
	\begin{enumerate}
		\item Set $S_1 = \mathrm{col}(\mf{X}), \alpha_1 = \Delta,$ $v_1 = \mf{0},$ and $g_1 = \rho^{\otimes m}$.  
		\item For $t = 1, 2, \ldots, T:$
		\begin{enumerate}
			\item Run $\textbf{PRDC}$ on the input data $S_{t}, v_{t}, \alpha_{t}, g_t$ to output $S_{t+1}, v_{t+1}$, and $\alpha_{t+1}$. 
			\item Set $g_{t+1}(x) = \frac{1}{\alpha_{t+1}} f(x/\alpha_{t+1})$, where $f(x)$ is the triangular density on $[-1,1]^m$. 
		\end{enumerate} 
		\item Apply \textbf{REDUCE} to the vectors in $S_T \cup \{ v_T \}$ to obtain $\sigma$. Let $v = \sum_{s_i \in S_T \cup \{ v_T \} } \sigma_i s_i$.
	\end{enumerate}
	
	\textbf{Output:} $|v|_\infty$
	\vspace{0.3 cm}
	
	
	
	We remark that the first three steps of \textbf{PRDC} are similar to those in the corresponding subroutine in \cite{KarKar82} for the one-dimensional case. The clean-up step and its analysis on the other hand are quite different. In particular, we use \textbf{REDUCE} to combine the `bad' vectors left over from resampling into a single bad vector $v\rp{0}$. This subroutine is quite similar to the algorithm used by Beck--Fiala to show that $t$-sparse vectors have discrepancy at most $2t - 1$ \citep{BecFia81}. In contrast, \cite{KarKar82} use a greedy iterative algorithm for dealing with bad points in dimension $1$, but it is not clear how to generalize their algorithm to also work in higher dimensions. In the next part of the clean-up step, we must bring the bad vector $v\rp{0}$ into a smaller range.  \cite{KarKar82} do this by randomly sampling points from $G'$ and greedily differencing them against $v\rp{0}$ until the resulting number is small. Here we use the same approach, but since we are working in higher dimensions, we measure the resulting vector in the Euclidean norm. In this part of the clean-up step, the key difference between our work and \cite{KarKar82} lies in our analysis, which includes elements of the analysis of stochastic gradient descent, as well as martingale concentration and the Khintchine inequality (see Appendix \ref{appendix:efficient_clean_up}). 

    We also comment on the reason for the bound $m = O(\sqrt{\log{n}})$ in Theorem \ref{thm:gkk}. First observe that by our choice of $\alpha = n^{-\Omega(1/m)}$ for the side-lengths of the sub-cubes at the first phase, it is necessary that $m = O(\log{n})$; otherwise the sub-cubes are not smaller than the original cube. The reason we require the stronger condition $m = O(\sqrt{\log{n}})$ is so that not too many points are labeled `bad' in the resampling step of our algorithm. We direct the reader to Appendix \ref{appendix:many_differences} for the analysis and further discussion.	

	\subsection{Analysis of GKK} 
	
	The proof of Theorem \ref{thm:gkk} follows from a sequence of inductive assumptions. Recall that $S_t$ denotes the points input to the $t^{\mathrm{th}}$ phase of \textbf{PRDC}, excluding the single `bad' vector $v_t \in \gamma m [-\alpha_t, \alpha_t]^m$, where $\gamma$ is a fixed absolute constant to be determined. Recall that $C^* = (2 \log(10/3))^{-1}$, as set in the definition of \textbf{GKK}, and that $\Delta > 0$ is the side length of the cube containing the initial set of vectors $S_1$.
	
	\begin{proposition}
		\label{prop:GKK_sizes_condl_distr}
		Let $X_1, \ldots, X_n$ be iid random vectors, each having a joint density $g:[-\Delta, \Delta]^m \to \mb{R}$. Consider the output $S_t, v_t, \alpha_t$ that results after the $(t-1)$-th phase of \textbf{\emph{PRDC}} in step 2 of \textbf{\emph{GKK}}. Then conditioned on $|S_j| = n_j$ for $1 \leq j \leq t$, we have
		\begin{itemize}
			\item the $n_t$ points in $S_t$ are iid and follow a triangular distribution on $[-\alpha_t, \alpha_t]^m$, and
			\item the random vector $v_t$ is independent of the vectors in $S_t$.
		\end{itemize} 
	\end{proposition}
	
	Proposition \ref{prop:GKK_sizes_condl_distr} ensures that the distribution of the output of each phase of recursion is preserved, allowing us to apply induction. At the heart of this result is the following marginal calculation which implies that the good points have a uniform distribution on their respective sub-cubes. Conditioning on $X_1 \in C_1$, if $L$ denotes the label of $X_1$ as `good' or `bad', then $(X_1, L)$ has a mixed joint density $p(x, \ell)$ where $x \in C_1$ and $\ell \in \{ `\text{good'}, \, `\text{bad'}\}$, which by Bayes' rule satisfies
	\[
	p(x|L=\text{`good'}) = \frac{ p(x, \, \text{`good'}) }{ \mb{P}[L = \text{`good'}] }
	= \frac{ g(x) \cdot \frac{ \min_{y \in C_1} g(y)}{g(x)} }{ \int_{C_1} p(y, \, \text{`good'}) \mr{d}y } = \frac{1}{\text{Vol}(C_1)},
	\]
	for all $x \in C_1$. 
	
	The proofs of Propositions \ref{prop:many_differences} and \ref{prop:efficient_clean_up} below are postponed to Appendices \ref{appendix:many_differences} and \ref{appendix:efficient_clean_up}, respectively. The former relies on showing that a large fraction of the points input to the $t^{\mr{th}}$ phase are labeled `good' in the \textbf{resample} step, and the latter requires us to show that few of the random differences created in step 3 of \textbf{PRDC} are lost in the \textbf{clean-up} step.
	
	
	\begin{proposition}
		\label{prop:many_differences}
		Suppose that $1 \leq t \leq C^* \log n$ and $m \leq C \sqrt{(\log n)/\max(1, \log \Delta)}$, where $C$ is a sufficiently small absolute constant. Then for some fixed $\theta$, conditioned on the events $|S_j| \geq \theta^{j-1} n$ for all $1 \leq j \leq t$, it holds that the set $G_{t}'$ of random differences created in step 2 of the $t^{\mathrm{th}}$ phase of \textbf{\emph{PRDC}} satisfies $|G_t'| \geq \beta |S_t|$ for some fixed $\beta$ with probability at least $1 - \exp(-c_1\sqrt{n})$, where $c_1>0$ is an absolute constant. In particular, we may set $\theta = 0.3$ and $\beta = 0.4$. 
	\end{proposition}
	
	\begin{proposition}
		\label{prop:efficient_clean_up}
		Suppose that $1 \leq t \leq C^* \log n$ and $m \leq C \sqrt{\log{n}}$, where $C$ is a sufficiently small absolute constant. Then conditioned on the events $|G_t'| \geq \beta |S_t|$ and $|S_j| \geq \theta^{j-1} n$ for $1 \leq j \leq t$, it holds that the set $S_{t+1}$ ( the input to the $(t+1)$-th iteration of \textbf{\emph{PRDC}}) satisfies $|S_{t+1}| \geq \theta |S_t|$ with probability at least $1 - \exp(-c_2 n^{1/4} )$, where $c_2>0$ is an absolute constant. In particular, we may choose $\beta = 0.4$ and $\theta = 0.3$. 
	\end{proposition}    

    The proof of Theorem \ref{thm:gkk} follows easily from the previous two propositions and is found in Appendix \ref{appendix:GKK_thm}.

    \appendix
    
    \section{Proofs from Section\texorpdfstring{ \MakeLowercase{\ref{sec:sub-linear_discrepancy}}}{} } 
    \label{appendix:sub-linear_discrepancy}
    
    First, we calculate the first and second moments of $S$ as defined in \eqref{eqn:defSvec}.   
    
    \begin{proof}[Proof of Lemma \ref{lem:moments}]
		Let $X_i^{(j)}$ denote the $j$th element of the vector $X_i$. Since these elements are independent, we get
		\begin{align*}
		\E[S] 
		&= \sum_{\sigma \in \{\pm 1 \}^n} \, \prod_{j = 1}^m \p\big(\big| \sum_{i = 1}^{n} \sigma_i X_i^{(j)} \big| \leq \eps \big)  = 2^n \p\big( |Z| \leq \frac{\eps}{\sqrt{n}} \big)^m 
		\end{align*}
		where $Z \sim \cN(0,1)$. This completes the proof of~\eqref{eqn:1mm}.

		To prove~\eqref{eqn:2mm}, let $d(\tau, \sigma)$ denotes the Hamming distance between $\sigma$ and $\tau$. Observe that if $\tau$ and $\sigma$ satisfy $d(\tau, \sigma) = k$, then $X:=\frac{1}{\sqrt{n}}\sum_{i = 1}^{n} \sigma_i X_i^{(j)}$ and $Y:=\frac{1}{\sqrt{n}}\sum_{i = 1}^{n} \tau_i X_i^{(j)}$ are $\rho_k$-correlated standard Gaussians  random variables. Thus
		\begin{align}
		\E[S^2] &= \sum_{\sigma, \tau \in \{\pm{1}\}^{n}} \p\big(\big|\sum_{i = 1}^{n} \sigma_i X_i\big|_\infty \leq \eps \, , \, \, \big|\sum_{i = 1}^{n} \tau_i X_i\big|_\infty \leq \eps  \big) \nonumber \\
		&= \sum_\sigma \sum_{k = 0}^{n} \, \sum_{\tau: \, d(\tau, \sigma) = k} \prok\big( \big|\sqrt{n}X\big| \leq \eps \, , \, \, \big|\sqrt{n}Y\big| \leq \eps  \big)^m \nonumber \\
		&= 2^n \sum_{k = 0}^n \binom{n}{k} \prok\left( \abs{\sqrt{n}X} \leq \eps \, , \, \, \abs{\sqrt{n}Y} \leq \eps  \right)^m \nonumber,
		\end{align}
		which proves the lemma. 
	\end{proof}
    
    The following small-ball probability estimates are required for the proof of the truncation argument, Lemma \ref{lem:truncation}. 
	
	\begin{lemma}
		\label{lem:small_ball_probs}
		Let $Z$ denote a standard Gaussian random variable, and let $X, Y$ denote $\rho$-correlated standard Gaussian random variables with $\rho \in (-0.5, 0.5)$. Then for $0 < z < 1$, we have for some absolute constant $c > 0$ that
		\begin{equation}
		-c z^3 \leq \p\big[ |Z| \leq z \big] - \sqrt{\frac{2}{\pi} } z  \leq 0, \label{eqn:gauss_approx} 
		\end{equation}
		and for all $z \in (0, \infty)$, we have
		\begin{equation}
		\p_\rho\big[ |X| \leq z, |Y| \leq z ] \leq \frac{2}{\pi \sqrt{1 - \rho^2}} z^2.  \label{eqn:rho_cor_approx}
		\end{equation} 
	\end{lemma}
    
    \begin{proof}
	    Observe that $z \mapsto \p[|Z| \leq z]$ is a concave function for $z \geq 0$. Hence, it lies below the tangent line to this curve at $z = 0$, which is precisely the function $z \mapsto \sqrt{2/\pi}z$. This proves the right-hand-side of \eqref{eqn:gauss_approx}. To prove the left-hand-side, we apply Taylor expansion and observe that for $|z| \leq 1$, it holds that
	    \[
	    \p[|Z| \leq z]  = \sqrt{\frac{2}{\pi}} z - \frac{1}{6} \sqrt{\frac{2}{\pi}} z^3 \pm O(z^5) \geq \sqrt{\frac{2}{\pi}} z - c z^3
	    \] 
	    for some absolute constant $c > 0$. 
	    
	    To prove \eqref{eqn:rho_cor_approx}, note that the joint density $\psi_\rho(x,y)$ of a pair of standard normal $\rho$-correlated Gaussians satisfies
	    \[
	    \psi_{\rho}(x,y)= \frac{1}{2\pi\sqrt{1-\rho^2}}\exp\big(-\frac{x^2-2\rho xy+y^2}{2-2\rho^2}\big) \leq \frac{1}{2\pi\sqrt{1-\rho^2}}.
	    \]
	    The upper bound follows by positive-semidefiniteness of the covariance matrix. Hence, integrating over the rectangle $|x| \leq z, |y| \leq z$ and applying the above upper bound yields the desired result.
	    
	\end{proof}
    
    \begin{lemma}
		\label{lem:truncation}
		Suppose that $\omega(1) = m = o(n)$ and let $\eps = \eps(n) = \gamma 2^{-n/m} \sqrt{\pi n/2} $ for some $\gamma > 1$. Then
		\begin{align}
		\label{eqn:truncation}
		2^n \sum_{k = 0}^{n/4} \binom{n}{k} \prok\left( \abs{\sqrt{n}X} \leq \eps \, , \, \, \abs{\sqrt{n}Y} \leq \eps  \right)^m = o( \E[S]^2 ). \\
		2^n \sum_{k = 3n/4}^{n} \binom{n}{k} \prok\left( \abs{\sqrt{n}X} \leq \eps \, , \, \, \abs{\sqrt{n}Y} \leq \eps  \right)^m = o( \E[S]^2 ). \label{eqn:truncation_sym}
		\end{align}
	\end{lemma}
    
    \begin{proof}
		Note that \eqref{eqn:truncation_sym} follows from \eqref{eqn:truncation} by symmetry, so it suffices to prove \eqref{eqn:truncation}. We may write $m = n/g_n$ for some sequence $g_n$ such that $\omega(1) = g_n = o(n)$. For notational convenience, define
		\[
		f_n(\rho) = \p_{\hspace{-.6ex}\rho}( \abs{\sqrt{n} X} \leq \eps, \abs{\sqrt{n} Y} \leq \eps ). 
		\] 
		By Lemma \ref{lem:moments}, we have
		\begin{multline}
		\label{eqn:AB}
		\frac{2^n \sum_{k = 0}^{n/4} \binom{n}{k} \prok\left( \abs{\sqrt{n}X} \leq \eps \, , \, \, \abs{\sqrt{n}Y} \leq \eps  \right)^m}{\mathbb{E}[S]^2} 
		\\= \underbrace{ \sum_{k = 0}^{n/(g_n)^2} \frac{\binom{n}{k}}{2^n} \left( \frac{ f_n(\rho_k) }{f_n(0) } \right)^m }_{=: A} + \underbrace{  \sum_{k = n/(g_n)^2}^{n/4} \frac{\binom{n}{k}}{2^n} \left( \frac{ f_n(\rho_k) }{f_n(0) } \right)^m }_{=: B}. 
		\end{multline}
		For $\eps$ as above and $Z \sim N(0,1)$, we have by applying \eqref{eqn:gauss_approx} that 
		\begin{equation}
		\label{eqn:first_moment_large}
		2^n \p(|Z| < \eps/\sqrt{n})^m \geq 2^n  \left( \sqrt{\frac{2}{\pi n} } \eps \right)^m (1 - c\eps^2/n )^m  \gtrsim_n \left( \frac{\gamma + 1}{2} \right)^m, 
		\end{equation}
		where $c$ is an absolute constant. To obtain the right-hand-side, note that $\eps/\sqrt{n} \xrightarrow{n \to \infty} 0$ since $m = o(n)$. Thus, for $n$ sufficiently large it holds that 
		\[
		1 - c\eps^2/n \geq \frac{1}{2}\left( 1 + \frac{1}{\gamma} \right),
		\]
		which yields the right-hand-side of \eqref{eqn:first_moment_large}. Now using the crude bound $f_n(\rho_k) \leq \p(|\sqrt{n}Z| \leq \eps)$, \eqref{eqn:first_moment_large}, the fact that $f_n(0) = \p(|\sqrt{n}Z| \leq \eps)^2$, and the inequality
		\[
		\sum_{k = 0}^j \binom{n}{k} \leq \left( \frac{ne}{j} \right)^j, 
		\]
		we have
		\begin{align}
		A &= \sum_{k = 0}^{n/(g_n)^2} \frac{\binom{n}{k}}{2^n} \left( \frac{ f_n(\rho_k) }{f_n(0) } \right)^{m} \nonumber \\
		&\lesssim_n 
		\left(\frac{\gamma + 1}{2}\right)^{-m} ( e \, g_n^2 )^{n/g_n^2} \nonumber \\
		&= \exp\left(- \frac{n \log \frac{1}{2}(1 + \gamma)}{g_n} + \frac{ n}{g_n^2} + \frac{2n \log g_n}{g_n^2} \right)  = o(1)  \label{eqn:A_bound} 
		\end{align} 
		because $(1/2)(1 + \gamma) > 1$, $g_n \to \infty$, and $n/g_n \to \infty$ as $n \to \infty$. 
		
		By \eqref{eqn:gauss_approx} and \eqref{eqn:rho_cor_approx} (noting again that $f_n(0) = \p(|\sqrt{n} Z| \leq \eps)^2$ ), we have
		\begin{align}
		\label{eqn:rho_cor_bound}
		B = \sum_{k = n/(g_n)^2}^{n/4} \frac{\binom{n}{k}}{2^n} \left( \frac{ f_n(\rho_k) }{f_n(0) } \right)^m \lesssim_n (c')^m \sum_{k = n/(g_n)^2}^{n/4} \frac{\binom{n}{k}}{2^n} \left( \frac{n^2}{k(n - k)} \right)^{m/2}  
		\end{align}
		where $c'$ is an absolute constant. By the Hoeffding bound, letting $c''$ denote another absolute constant, we have
		\begin{align*}
		\eqref{eqn:rho_cor_bound} \lesssim_n (c'')^m g_n^m e^{-n/8} = 
		\exp\left( \frac{n \log c''}{g_n} + \frac{n \log g_n}{g_n} - \frac{n}{8} \right) = o(1)
		\end{align*}
		since $g_n \to \infty$. Since $A, B = o(1)$, we conclude by \eqref{eqn:AB} that \eqref{eqn:truncation} holds, as desired.
		
	\end{proof}
    
    \begin{lemma}
		\label{lem:logconcave}
		Suppose that $m=o(n)$ and set $\eps=\gamma2^{-n/m}\sqrt{n\pi/2}$. Then the function $\alpha \mapsto \phi_n(\alpha)$ defined in~\eqref{eqn:defphi_n} is asymptotically strictly concave on $(0.25,0.75)$. More precisely,
		$$
		\lim_{n \to \infty}\frac1n \frac{\partial^2}{\partial \alpha^2}\phi_n(\alpha)= - \frac{1}{\alpha(1 - \alpha)} < -4\,, \qquad \forall \,\alpha \in (0.25, 0.75)\, ,
		$$
		and the convergence is uniform over $\alpha \in (0.25, 0.75)$. Moreover, for $n$ large enough, $\phi_n(\alpha)$ has a unique maximum over $(0.25, 0.75)$ located at $\alpha = 0.5$.  
	\end{lemma}
    
    \begin{proof} 
		Because $|\partial_\alpha^2 \log \alpha(1 - \alpha)| = O(1)$ for $\alpha \in (0.25, 0.75),$ $m = o(n)$, and 
		\[
		h''(\alpha) = - \frac{1}{\alpha(1 - \alpha)},
		\]
		to verify the strict concavity of $\phi_n(\alpha)$, it suffices to show that
		\begin{equation}
		\label{eqn:cor2deriv_limit}
		\abs{ \frac{\partial^2}{\partial \alpha^2}  \log \p_{1 - 2\alpha}\left[ \abs{\sqrt{n}X} \leq \eps \, , \, \, \abs{\sqrt{n}Y} \leq \eps  \right]} = O(1), \quad \alpha \in (0.25, 0.75).
		\end{equation}
		
		For notational convenience, we write $f_n(\rho)=\p_{\hspace{-.5ex}\rho}\left( |\sqrt{n}X| \leq \eps \, , \, \, |\sqrt{n}Y| \leq \eps  \right)$. We study the logarithmic second derivative
		\begin{equation}
		\label{EQ:defJn}
		J_n(\rho):=\frac{f_n''(\rho)}{f_n(\rho)}-\big(\frac{f'_n(\rho)}{f_n(\rho)}\big)^2
		\end{equation}
		by controlling each term individually.
		
		First, recall that for any $\rho \in (-1,1)$, the distribution $\p_{\rho}$ admits a density with respect to the Lebesgue measure over $\mathbb{R}^2$ given by
		$$
		\psi_{\rho}(x,y)=\frac{1}{2\pi\sqrt{1-\rho^2}}\exp\big(-\frac{x^2-2\rho xy+y^2}{2-2\rho^2}\big).
		$$
		It holds that
		$$
		f_n'(\rho)=\iint_{[-\frac{\eps}{\sqrt{n}},\frac{\eps}{\sqrt{n}}]^2}\drho\psi_{\rho}(x,y)\ud x\ud y.
		$$
		Thus since $\eps=o(\sqrt{n})$ we get,
		\begin{equation}
		\label{EQ:f'/f}
		\lim_{n \to \infty}\frac{f'_n(\rho)}{f_n(\rho)}=\lim_{n \to \infty}\frac{\frac{\eps^2}{n}\iint_{[-\frac{\eps}{\sqrt{n}},\frac{\eps}{\sqrt{n}}]^2}\partial_\rho\psi_{\rho}(x,y)\ud x\ud y}{\frac{\eps^2}{n}\iint_{[-\frac{\eps}{\sqrt{n}},\frac{\eps}{\sqrt{n}}]^2}\psi_{\rho}(x,y)\ud x\ud y}=\frac{\drho \psi_{\rho}(0,0)}{\psi_{\rho}(0,0)}=\drho \log(\psi_\rho)(0,0).
		\end{equation}
		Similarly,
		\begin{equation}
		\label{eqn:f''/f}
		\lim_{n \to \infty}\frac{f''_n(\rho)}{f_n(\rho)}=\frac{\drho^2 \psi_{\rho}(0,0)}{\psi_{\rho}(0,0)}=\drho^2 \log(\psi_\rho)(0,0)+\big(\drho \log(\psi_\rho)(0,0)\big)^2.
		\end{equation}
		Together with~\eqref{EQ:defJn} and~\eqref{EQ:f'/f}, the above display yields
		$$
		\lim_{n \to \infty}J_n(\rho)=\frac{1+\rho^2}{(1-\rho^2)^2} = O(1), 
		$$
		if $\rho \in (-0.5, 0.5)$. Moreover, the convergence in \eqref{EQ:f'/f} and\eqref{eqn:f''/f} is uniform over $\rho \in (-0.5, 0.5)$. This is because the functions $\psi_\rho, \partial_\rho \psi_\rho,$ and $\partial^2_\rho \psi_\rho$ are all $C$-Lipschitz on $\mathbb{R}^2$ for some absolute constant $C > 0$, provided that we restrict $\rho \in (-0.5, 0.5)$. Next, changing variables via $\rho = 1- 2 \alpha$, this verifies \eqref{eqn:cor2deriv_limit}. Thus we have shown that $\phi_n(\alpha)$ is strictly concave on $(0.25, 0.75)$ for $n$ sufficiently large, completing the first part of the proof.
		
		The strict concavity verifies that $\phi_n(\alpha)$ has a unique maximum on $(0.25, 0.75)$. We show that it occurs at $\alpha = 0.5$. It is easy to check that both $h(\alpha)$ and $\alpha \mapsto \log \frac{1}{\sqrt{\alpha(1 - \alpha)}}$ have a critical point at $\alpha = 1/2$. So, applying the change of variables $\rho = 1 - 2 \alpha$, we just need to verify that $f_n'(0) = 0$. Let $\phi(x) = \frac{1}{\sqrt{2 \pi}} e^{-x^2/2}$ denote the density of a standard Gaussian and set $\ell = \eps/\sqrt{n}$. Straightforward calculus shows that 
		\[
		\frac{\partial}{\partial \rho}\bigg|_{\rho = 0} \psi_\rho(x,y) = x y \phi(x) \phi(y).  
		\]
		Therefore,
		\[
		\frac{\partial}{\partial \rho}\bigg|_{\rho = 0} f_n(\rho) = 
		\left( \int_{-\ell}^{\ell} x \phi(x) \right)^2 = 0.
		\]
		This proves the second part of the lemma, so we're done.
	\end{proof}	
	

	\section{Gaussian discrepancy in small linear dimension}
	\label{appendix:small_linear}
	
	The goal of this appendix is to prove the result below, which combined with Theorem \ref{thm:sub-linear_discrepancy} and Theorem 2 of \cite{ChaVem14} provides a precise characterization of asymptotic Gaussian discrepancy. 
	
	\begin{theorem}
		\label{thm:small_linear_discrepancy}
		Let $X_1, \ldots, X_{n} \stackrel{iid}{\sim} \cN(0,I_m)$ be independent standard Gaussian random vectors. Let $\gamma > 1$ denote an arbitrary absolute constant. Then there exists $\Delta = \Delta(\gamma)$ such that for $m \leq \Delta n$, 
		\begin{equation}	
		\label{eqn:linear_upper_bound}
		\liminf_{n \to \infty}\p\Big[\mc{D}(X_1, \ldots, X_n) \leq\gamma\sqrt{\frac{\pi n}{2}}2^{-n/m}\Big] \geq 0.99 \,.
		\end{equation}
		
	\end{theorem} 
	
	In particular, combining Theorem \ref{thm:small_linear_discrepancy} with Theorem 2 of \cite{ChaVem14}, we can now estimate the discrepancy up to constant factor, with probability asymptotically larger than $99 \%$, in the entire linear regime $m = \delta n$ where $\delta > 0$. Note that our guarantee on the probability here is weaker than that of the high-probability upper bound from Theorem \ref{thm:sub-linear_discrepancy}. The constant $0.99$ can be boosted to be arbitrarily close to $1$ by choosing smaller $\Delta$, though our techniques do not allow us to set the right-hand-side to be $1$ for any fixed $\Delta > 0$. 
	
	The closely related work of \cite{AubPerZde19} also considered Gaussian discrepancy in the linear regime $m = \delta n$ for fixed $\delta > 0$. Subject to a certain numerical hypothesis, the authors showed that
	\begin{equation}
	\label{eqn:AubPerZde_upper_bound}
	\liminf_{n \to \infty} \p\left[ \mc{D}(X_1, \ldots, X_n) \leq c(\delta) \sqrt{n} \right]   > 0,
	\end{equation}
	where $c(\delta)$, as a function of $\delta$, is the inverse of the function $x \mapsto \log(1/2)/\p[|Z| \leq x]$ and $Z \sim N(0,1)$. Their proof is an application of the second moment method, similar to ours. They also showed the following high-probability lower bound using the first moment method:
	\begin{equation}
	\label{eqn:AubPerZde_lower_bound}
	\lim_{n \to \infty} \p\left[ \mc{D}(X_1, \ldots, X_n) \geq (c(\delta) - \eps)\sqrt{n} \right]   = 1,
	\end{equation}
	where $\eps >0$ is an arbitrary absolute constant. \cite{AubPerZde19} conjectures, with strong evidence using heuristics from statistical mechanics, that the event in \eqref{eqn:AubPerZde_upper_bound} holds with probability tending to $1$. We remark that as $\delta \to 0$, we have $c(\delta)  = \Theta(2^{-1/\delta}) = \Theta(2^{-n/m})$. Theorem \ref{thm:small_linear_discrepancy} shows that with a constant factor's worth of `extra room' in the discrepancy threshold, the asymptotic probability in \eqref{eqn:AubPerZde_upper_bound} can be boosted to be arbitrarily close to $1$.  
	
	On the algorithmic side, using a mild extension of the techniques of \cite{ChaVem14}, in dimension $m = \delta n $ with $\delta \in (0,1)$, one can show an algorithmic bound of $O(\sqrt{ \delta n })$ on the discrepancy, and this is the best known result for this regime. Hence, Theorem \ref{thm:small_linear_discrepancy} suggests the possibility of a statistical-to-computational gap in the small linear regime $m = \delta n$ for $\delta \in (0, 1)$. Note that for $\delta > 1$, the results of \cite{ChaVem14} confirm an absence of statistical-to-computational gaps in the discrepancy.  
	
	
	The proof of Theorem \ref{thm:small_linear_discrepancy} follows closely the steps from Section \ref{sec:sub-linear_discrepancy} with some modifications. We begin with a truncation argument as in Lemma \ref{lem:truncation}.
	
	\begin{lemma}
		\label{lem:linear_truncation}
		Let $\gamma > 1$ denote an arbitrary absolute constant. Then there exists $\Delta = \Delta(\gamma)$ such that if $m = \delta n$ for $\delta \leq \Delta$ and $\eps = \eps(n) = \gamma 2^{-1/\delta} \sqrt{\pi n/2} $, then
		
		\begin{align}
		\label{eqn:linear_truncation}
		2^n \sum_{k = 0}^{n/4} \binom{n}{k} \prok\left( \abs{\sqrt{n}X} \leq \eps \, , \, \, \abs{\sqrt{n}Y} \leq \eps  \right)^m = o( \E[S]^2 ). \\
		2^n \sum_{k = 3n/4}^{n} \binom{n}{k} \prok\left( \abs{\sqrt{n}X} \leq \eps \, , \, \, \abs{\sqrt{n}Y} \leq \eps  \right)^m = o( \E[S]^2 ). \label{eqn:linear_truncation_sym}
		\end{align}
	\end{lemma}

	\begin{proof}
		
%
		The proof follows closely that of Lemma \ref{lem:truncation}, setting $g_n \equiv 1/\delta$. We set 
		\[
		f_\delta(\rho) = \p_{\hspace{-.6ex}\rho}( \abs{\sqrt{n} X} \leq \eps, \abs{\sqrt{n} Y} \leq \eps ) = \p_{\hspace{-.6ex}\rho}( \abs{X} \leq \gamma 2^{-1/\delta} \sqrt{\pi/2} , \, \, \abs{ Y} \leq \gamma 2^{-1/\delta} \sqrt{\pi/2}  ). 
		\] 
		Note that the function $f_\delta$ is independent of $n$ by our choice of $\eps$. As in \eqref{eqn:AB} from Lemma \ref{lem:truncation}, we let 
		\[
		A = \sum_{k = 0}^{\delta^2 n} \frac{\binom{n}{k}}{2^n} \left( \frac{ f_\delta(\rho_k) }{f_\delta(0) } \right)^m, \quad B =  \sum_{k = \delta^2 n}^{n/4} \frac{\binom{n}{k}}{2^n} \left( \frac{ f_\delta(\rho_k) }{f_\delta(0) } \right)^m.
		\]
		
		Note that for $\delta$ sufficiently small (depending on $\gamma$), it holds that $\eps/\sqrt{n} \leq 1$. Therefore, similar to \eqref{eqn:first_moment_large}, we can apply the lower bound from Lemma \ref{lem:small_ball_probs} to conclude that
		\begin{equation}
		\label{eqn:first_moment_linear}
		2^n \p[|Z| < \eps/\sqrt{n}]^m \geq 2^n\left( \sqrt{\frac{2}{\pi n}} \eps  \right)^m(1 - c \eps^2/n)^m \geq \left( \frac{\gamma +1}{2} \right)^m,
		\end{equation} 
		Hence, as in \eqref{eqn:A_bound} we have
		\begin{align}
		A 
		\lesssim_n 
		\left(\frac{\gamma + 1}{2}\right)^{-m} ( e \delta^{-2} )^{\delta^2 n} = \exp\left(- \delta n \log\left(\frac{1}{2}(1 + \gamma)\right) + \delta^2 n + 2 \delta^2 n \log(1/\delta)  \right).  
		\end{align} 
		Hence, if $\delta \leq \Delta(\gamma)$ for $\Delta(\gamma)$ sufficiently small, then we have that $A = o(1)$.
		
		Similar to \eqref{eqn:rho_cor_bound}, we have by applying \eqref{eqn:gauss_approx} and \eqref{eqn:rho_cor_approx} that 
		\begin{equation}
		\label{eqn:linear_cor_bound}
		B \lesssim_n (c'(\gamma))^m \sum_{k = \delta^2 n}^{n/4} \frac{\binom{n}{k}}{2^n} \left( \frac{n^2}{k(n-k)} \right)^{m/2}.
		\end{equation}
		By the Hoeffding bound (letting $c''(\gamma)$ denote another constant depending on $\gamma$), 
		we have
		\begin{equation}
		\eqref{eqn:linear_cor_bound} \lesssim_n (c''(\gamma))^m \delta^{-m} e^{-n/8} = \exp\left( \delta n \log(c''(\gamma)) + \delta n \log(1/\delta)  - n/8\right) = o(1),
		\end{equation}
		provided that $\delta \leq \Delta(\gamma)$ for $\Delta(\gamma)$ sufficiently small. Since $A = o(1)$ as well for this setting of parameters, the lemma follows. 
	\end{proof}

	Our next lemma is a version of Lemma \ref{lem:logconcave} corresponding to the linear regime. We use the log-concavity of the function $\phi_n$ when we apply the Laplace method to the second moment, as in the sub-linear regime.   

	\begin{lemma}
		\label{lem:linear_logconcave}
		Let $\eta > 0$ and $\gamma > 1$ be arbitrary constants, and let $\Delta = \Delta(\gamma, \eta)$ denote a sufficiently small absolute constant. Suppose that $m = \delta n$ for $\delta \leq \Delta$, and set $\eps = \gamma 2^{-1/\delta} \sqrt{n \pi/2}$. Then the function $\alpha \mapsto \phi_n(\alpha)$ defined in \eqref{eqn:defphi_n} is strictly concave on $(0.25, 0.75)$. More precisely, 
		\begin{equation}
		\label{eqn:phi_deriv_linear}
		\frac1n \frac{\partial^2}{\partial \alpha^2}\phi_n(\alpha) \leq - \frac{1}{\alpha(1 - \alpha)} + \eta < -4 + \eta \,, \qquad \forall \,\alpha \in (0.25, 0.75).
		\end{equation}
		Moreover, $\phi_n(\alpha)$ has a unique maximum over $(0.25, 0.75)$ located at $\alpha = 0.5$.
	\end{lemma}
	
	\begin{proof}
		Recall that
		\[
		f_\delta(\rho) = \p_{\hspace{-.6ex}\rho}( \abs{X} \leq \gamma 2^{-1/\delta} \sqrt{\pi/2} , \, \, \abs{ Y} \leq \gamma 2^{-1/\delta} \sqrt{\pi/2}  ). 
		\] 
		As in the proof of Lemma \ref{lem:logconcave}, it suffices to study the logarithmic second derivative with respect to $\rho$
		\begin{equation}
		\label{eqn:def_linear_Jn}
		J_{\delta}(\rho):=\frac{f_\delta''(\rho)}{f_\delta(\rho)}-\big(\frac{f'_\delta(\rho)}{f_\delta(\rho)}\big)^2
		\end{equation} 
		and show that $|J_\delta(\rho)| = O(1)$ for $\rho \in (-0.5, 0.5)$. Recall that $\psi_\rho$ denotes the density associated to $\p_{\hspace{-.6ex}\rho}$.
		
		Since $\eps/\sqrt{n} \to 0$ as $\delta \to 0$, we have, similar to \eqref{EQ:f'/f}, that
		\begin{equation}
		\label{EQ:linear_f'/f}
		\lim_{\delta \to 0}\frac{f'_\delta(\rho)}{f_\delta(\rho)}=\lim_{\delta \to 0}\frac{\frac{\eps^2}{n}\iint_{[-\frac{\eps}{\sqrt{n}},\frac{\eps}{\sqrt{n}}]^2}\partial_\rho\psi_{\rho}(x,y)\ud x\ud y}{\frac{\eps^2}{n}\iint_{[-\frac{\eps}{\sqrt{n}},\frac{\eps}{\sqrt{n}}]^2}\psi_{\rho}(x,y)\ud x\ud y}=\frac{\drho \psi_{\rho}(0,0)}{\psi_{\rho}(0,0)}=\drho \log(\psi_\rho)(0,0).
		\end{equation}
		And similar to \eqref{eqn:f''/f}, we have
		\begin{equation}
		\label{eqn:linear_f''/f}
		\lim_{\delta \to 0}\frac{f''_\delta(\rho)}{f_\delta(\rho)}=\frac{\drho^2 \psi_{\rho}(0,0)}{\psi_{\rho}(0,0)}=\drho^2 \log(\psi_\rho)(0,0)+\big(\drho \log(\psi_\rho)(0,0)\big)^2.
		\end{equation}
		It follows that
		\[
		\lim_{\delta \to 0} J_\delta(\rho) = \frac{1 + \rho^2}{(1 - \rho^2)^2} = O(1)
		\]
		for $\rho \in (-0.5, 0.5)$. Moreover, similar to the proof of Lemma \ref{lem:logconcave}, the convergence in \eqref{EQ:linear_f'/f} and \eqref{eqn:linear_f''/f} is uniform in $\delta$ by the Lipschitzness of $\psi_\rho, \partial_\rho \psi_\rho$, and $\partial_\rho^2 \psi_\rho$ over the interval $\rho \in (-0.5, 0.5)$. Therefore, if we take $\delta$ sufficiently small with respect to $\gamma, \eta$, then \eqref{eqn:phi_deriv_linear} holds. 
		
		Note that independent of $\eps$, we have that $\rho = 0$ is a critical point of $\phi_n$, as shown at the end of the proof of Lemma \ref{lem:logconcave}. Applying this and making the change of variables $\rho = 1 - 2\alpha$ verifies the last statement of Lemma \ref{lem:linear_logconcave}.
	\end{proof}
		
	\begin{proof}[Proof of Theorem \ref{thm:small_linear_discrepancy}]
		Recall from the definition in \eqref{eqn:main_term} that
		\[
		L:=2^n \sum_{k = n/4}^{3n/4} \binom{n}{k} \prok\left( \abs{\sqrt{n}X} \leq \eps \, , \, \, \abs{\sqrt{n}Y} \leq \eps  \right)^m.
		\]
		Applying Stirling's formula and a Riemann sum approximation as in \eqref{eqn:Stirling} and \eqref{eqn:Riemann_sub-linear}, respectively, we have that
		\begin{equation}
		\label{eqn:linear_Stirling}
		L \lesssim_n 2^n \sqrt{\frac{n}{2 \pi}} \int_{1/4}^{3/4} \exp( \phi_n(\alpha) ) d\alpha. 
		\end{equation}
		Since $ \phi_n(\alpha)/n$ is independent of $n$, we can apply the Laplace method directly \citep[see][]{ murray} along with Lemma \ref{lem:linear_logconcave} to see that
		\begin{equation}
		\label{eqn:linear_Laplace}
		\int_{1/4}^{3/4} \exp( \phi_n(\alpha) ) d\alpha \lesssim_n \sqrt{ \frac{2 \pi}{|\phi_n''(1/2)|}} \exp( \phi_n(1/2) ) \leq \sqrt{\frac{2 \pi}{n(4 - \eta)}} 2^{n+1} f_\delta(0)^m.  
		\end{equation}
		assuming $\delta \leq \Delta$ for $\Delta(\gamma, \eta)$ sufficiently small. 
		
		Therefore, by Lemma \ref{lem:truncation}, \eqref{eqn:linear_Stirling}, \eqref{eqn:linear_Laplace}, Lemma \ref{lem:moments}, the definition of $f_\delta$, and assuming that $\delta \leq \Delta$ for $\Delta(\gamma, \eta)$ sufficiently small, we have
		\[
		\mathbb{E}[S^2] \lesssim_n L \lesssim_n \sqrt{\frac{4}{4 - \eta}} (2^{n} \p[|\sqrt{n} Z| \leq \eps ]^m)^2  = \sqrt{\frac{4}{4 - \eta}} \mathbb{E}[S]^2. 
		\]
		Setting $\eta = 10^{-5}$, we have by the second moment method \eqref{eqn:second_moment} that
		\[
		\p[S > 0] \geq \frac{\mathbb{E}[S]^2}{\mathbb{E}[S^2]} \gtrsim_n \sqrt{1 - \eta/4} \geq 0.99,
		\]
		completing the proof of Theorem \ref{thm:small_linear_discrepancy}. 
	\end{proof} 
	
	\section{The REDUCE algorithm} 
    \label{appendix:reduce} 

    In this appendix we define the \textbf{REDUCE} algorithm, a simple procedure for combining a set of points into a single point whose $\ell_\infty$-norm is not too large.  This algorithm \textbf{REDUCE} is described explicitly below, and its main property of use is described in Lemma \ref{lem:reduce}, whose proof is given below. The analysis of this algorithm uses feasibility as in the classical proof of the Beck-Fiala theorem \citep{alonspencer}. 
	
	\vspace{0.3 cm}
	\noindent \textbf{REDUCE:}
	
	\textbf{Input}: $m \times N$ matrix $\mf{X}$ with columns $X_1, \ldots, X_N$. 
	
	If $N < m$: 
	
	\indent \indent Choose $s \in \{\pm 1\}^N$ arbitrarily. 
	
	Else:
	\begin{enumerate}
		\item Let $s\rp{0} = \mf{0} \in \mb{R}^N$, and let $T_0 = \emptyset$.
		\item For $k = 0, 1, 2, \ldots $
		
		\indent If $|T_{k}| < N - m$
		\begin{enumerate}
			\item Find (e.g., using Gaussian elimination) a vector $v \neq \mf{0} \in \mb{R}^N$ such that $\mf{X} v = \mf{0}$ and $v_j = 0$ for all $j \in T_k$. 
			\item Define $s\rp{k+1} = s\rp{k} + \lambda v$, where $\lambda>0$ is the smallest real number such that $|s\rp{k}_j + \lambda v_j| = 1$ for some $j \notin T_k$. 
			\item Define $T_{k+1} = \{ j: \abs{s\rp{k+1}} = 1 \}$.  
		\end{enumerate}
		Else: $s := s\rp{k}$. BREAK
	\end{enumerate} 
	
	\textbf{Output:} $\sigma := \sgn(s)$
	\vspace{0.3 cm}

    \begin{proof}[Proof of Lemma \ref{lem:reduce}] 
		We suppose that $N > m$, otherwise, an arbitrary choice of signing gives the desired upper bound. 
		Suppose that we are in the $k$-th iteration of Step 2 of \textbf{REDUCE}. If $|T_k| < N - m$, then there are at most $m + |T_k| < N$ linear constraints on the vector $v \in \mathbb{R}^N$ in step 2(a). So by dimension-counting, there exists a nonempty subspace of feasible $v$. Next if $s\rp{k} \in [-1,1]^m$, then $\lambda$ from step 2(b) exists and furthermore $s\rp{k+1} \in [-1,1]^m$ by the choice of $j$ in step 2(b). Also, we have that $T_k \subset T_{k+1}$; if $|(s\rp{k})_j| = 1$, then the $j$-th coordinate remains unchanged for future iterations of step 2. Finally, $|T_k|$ increases at least by $1$ in each iteration, so the loop in step 2 is guaranteed to terminate after at most $N - m$ iterations. 
		
		It remains to verify that $\sigma$ satisfies the upper bound from Lemma \ref{lem:reduce}. Observe that $s \in [-1,1]^m$, $T := |\{j: \abs{s_j} = 1\}| \geq N - m$, and
		\[
		\sum_{i= 1}^N  s_i X_i = \mf{0}. 
		\] Therefore,
		\begin{align*}
		\abs{\sum_{i = 1}^N \sigma_i X_i}_\infty &\leq 
		\abs{\sum_{i = 1}^N s_i X_i}_\infty + \abs{\sum_{i \notin T} (\sgn(s_i) - s_i) X_i}_\infty  \\
		&\leq \max_{ S \subset [N]: |S| = m } \sum_{i \in S} \abs{X_i}_\infty. 
		\end{align*}  
	\end{proof}
	
	\section{Proof of Proposition\texorpdfstring{ \MakeLowercase{\ref{prop:many_differences}}}{} }
	\label{appendix:many_differences}
	
    We need to show that at each application of resampling in \textbf{GKK}, a small number of points are labeled `bad'. As discussed in the introduction, the restriction on the dimension $m = O(\sqrt{\log{n}})$ is needed in our analysis to show that the probability of a point being labeled `bad' is small. 

	We briefly describe the intuition for this condition by considering the first phase of the algorithm \textbf{GKK}. Suppose, for example, that $X_1, \ldots, X_n$ are independent triangularly distributed vectors on $[-1, 1]^m$. In step 1 of \textbf{PRDC}, the cube $[-1,1]^m$ is partitioned into sub-cubes of side length $\alpha' = n^{-\Omega(1/m)}$. Next, we enter the resampling step. We show below that the probability of a point being labeled `bad' is at most $O( 2^m m \alpha') = O(  2^m m n^{-\Omega(1/m)})$. Roughly speaking, the reason for this is that there are $2^m (\alpha')^{-m}$ sub-cubes, and the probability of a point in a particular sub-cube being labeled `bad' is controlled by the product of three terms: 1) the $\ell_1$-Lipschitz constant of the density of $X_1$, which is $1$, 2) the $\ell_1$-diameter of the sub-cube, which is $m \alpha'$, and 3) the volume of the sub-cube, which is $(\alpha')^m$. Hence, the probability of a point being labeled `bad' is a small constant, assuming that $m = O(\sqrt{\log{n}})$. 
	
	The next two lemmas present the above argument in full detail. 

	
	
	\begin{lemma}
		\label{lem:Lipschitz_density}
		Let $\rho:[-\Delta,\Delta] \to \mb{R}$ denote a pdf that is $L$-Lipschitz and bounded above by some constant $D > 0$. Let $g = \rho^{\otimes m}:[-\Delta,\Delta]^m \to \mb{R}$ denote the density of the distribution of $m$ independent random variables, each individually distributed according to $\rho$. Then $g$ is $L'$-Lipschitz in the $\ell_1$ norm:
		\[
		\forall \, x, y \in [-\Delta, \Delta]^m, \quad |g(x) - g(y)| \leq L' \abs{x - y}_1, 
		\]
		where 
		\[
		L' = L D^{m-1} .
		\]
	\end{lemma}
	
	\begin{proof}
		
	    Define $x^1 = x$, and for $2 \leq k \leq m$, define
		\[
		x^k = x^{k-1} + \mf{e}_k(y_k - x_k),
		\]
		where $\mf{e}_k$ denotes the $k$-th elementary basis vector. Then we have
		\begin{align*}
		\abs{g(x) - g(y)} &\leq \sum_{k = 1}^m \abs{ g(x^{k}) - g(x^{k-1}) } 
		\left(\prod_{i < k} g(y_i)\right) \left(  \prod_{i > k} g(x_i) \right) 
		\\
		&\leq \sum_{k = 1}^m L D^{m-1} |x_k - y_k| \\
		&= L D^{m-1} \abs{x - y}_1.
		\end{align*}
	\end{proof}
	
	\begin{lemma}
		\label{lem:lost_bad_points}
		Let $S = X_1, \ldots, X_s \in [-\Delta,\Delta]^m$ denote a sample of iid random vectors, each having a joint density $g = \rho^{\otimes m}$, where $\rho$ is $L$-Lipschitz and bounded above by $D > 0$.  Let $B$ denote the bad points created in step 2 of \textbf{\emph{PRDC}} run on the input $S, v = 0, \alpha = \Delta,$ and $g$. If $m \leq C \sqrt{\log(s)/\max(1, \log \Delta)}$ for a sufficiently small constant $C = C(D, L) > 0$, then 
		\[
		\p[ |B| > 0.1 s ] \leq \exp(-c_1 s),
		\]
		where $c_1$ is an absolute constant.
	\end{lemma} 
	
	\begin{proof}
		Let $\alpha' = \Delta/\ceil{s^{1/(4m)}}$. Let $C_1, \ldots, C_N$ denote the sub-cubes of side length $\alpha'$ formed by partitioning (step 1 of \textbf{PRDC}), recalling that $N = (2\Delta)^m (\alpha')^{-m}$. Since $X_1, \ldots, X_s$ are independent, we first study the probability that $X_1$ is bad and then apply a Hoeffding bound.
		\begin{align*}
		\p[ X_1 \, \text{is bad} ] &= \sum_{j = 1}^N \int_{C_j} \left(1 - \frac{\min_{y \in C_j} g(y)}{g(x)}\right) g(x) \, dx \\
		&= \sum_{j = 1}^N \int_{C_j} \left( g(x) - \min_{y \in C_j} g(y) \right) \,dx \\
		&\leq \sum_{j = 1}^N \mr{Vol}(C_j) L D^{m-1} \, \mr{diam}_{\ell_1}(C_j) \\
		&= (2 \Delta)^m L D^{m-1} m \alpha',
		\end{align*}
		where we measure the diameter in the $\ell_1$ norm and applied Lemma \ref{lem:Lipschitz_density}. Since \[m \leq C \sqrt{\log(s)/\max(1, \log \Delta)},\] we have
		\[
		p := (2 \Delta)^m L D^{m-1} m \alpha' \leq (2 \Delta)^m D^{m-1} m \Delta s^{-1/(4m)} \leq 0.05
		\]
		for $C = C(D, L) > 0$ sufficiently small. Since the $X_i$'s are independent, by Hoeffding's inequality,
		\[
		\p[ \, |B| \geq 0.1s] \leq  \p[ \, |B| - ps \geq   0.05s] \leq 
		\exp\left(-\frac{2(0.05)^2s^2}{s}\right),
		\]
		which completes the proof.
	\end{proof}
%
	
	\begin{proof}[Proof of Proposition \ref{prop:many_differences}]
		The proof is by induction on $t$. We first handle the base case $t =1$. By assumption the matrix $\mf{X}$ has independent entries, each having a pdf which is $L$-Lipschitz and bounded above by $D$. By Lemma \ref{lem:lost_bad_points}, with probability at least $1 - \exp(-c_1n)$, there are at most $0.1 n$ points labeled `bad'. Since $m \leq C \sqrt{\log(n)/\max(1, \log \Delta)}$, for $C$ sufficiently small, there are at most $N_1 \leq (2 \Delta)^m \alpha_2^{-m} \leq n^{0.6} $ sub-cubes created by partitioning (step 1 of \textbf{PRDC}). Thus, at most that many good points are leftover after random differencing in step 3 of \textbf{PRDC}. We conclude that with probability at least $1 -\exp(-c_1n)$, there are at least 
		\begin{equation}
		\label{eqn:halving}
		\frac{n - 0.01n - n^{0.6}}{2} \geq 0.4 n
		\end{equation}
		points in $G_1'$, the set of random differences. 
		
		Now we show the inductive step. Let $\mc{E}$ denote the event $|S_j| = n_j$ where $n_j \geq (0.3)^{j-1} n$ for all $1 \leq j \leq t$. It suffices to show that
		\begin{equation}
		\label{eq:many_differences}
		\p\left[ \, |G'_{t+1}| \leq 0.4 n_t \, \bigg| \, \mc{E} \right] \leq \exp\left( - c_1 \sqrt{n} \right).
		\end{equation}
		 By Proposition \ref{prop:GKK_sizes_condl_distr} in Appendix \ref{app:distribution_properties}, conditionally on $\mc{E}$, the distribution of the points in $S_t= \y_1, \ldots, \y_{n_t}$ are iid and follow a triangular distribution on $[-\alpha_t, \alpha_t]^m$. Hence, we have by Lemma \ref{lem:Lipschitz_density} that the density of $\alpha^{-1}_t\y_1, \ldots, \alpha^{-1}_t\y_{n_t}$ is $1$-Lipschitz with respect to $\ell_1$ and is bounded above by $D= 1$. Note that, by an application of the chain rule, the probability $\alpha^{-1}_t\y_j$ is labeled `good' using the triangular density on $[-1, 1]^m$ for $g$ in step 2 of \textbf{PRDC} is the same as the probability that $\y_j$ is labeled `good' using the triangular density on $[-\alpha_t, \alpha_t]^m$ for $g$ in step 2 of \textbf{PRDC}. 
		
		Since $t \leq \ceil{C^* \log n}$ and $n_j \geq (0.3)^{j-1} n$ for $1 \leq j \leq t$, we have that $n_t \geq \sqrt{n}$. In particular, for $C>0$ sufficiently small, $s = \sqrt{n}$ satisfies the required lower bound of Lemma \ref{lem:lost_bad_points}. Therefore, 
		\[
		\p\left[\, |B_{t+1}| \geq 0.1n_t \, \bigg| \, \mc{E}\right] \leq \exp(-c_1 n_t) \leq \exp(-c_1 \sqrt{n}). 
		\]  
		For $C$ sufficiently small and $m \leq C \sqrt{\log n}$, there are at most $N_t \leq 2^m n_t^{1/4} \leq n_t^{0.6}$ sub-cubes formed in step 1 of \textbf{PRDC}. Hence, at most $n_t^{0.6}$ good points are leftover after the random differencing step of \textbf{PRDC}. Halving the number of remaining points as in \eqref{eqn:halving} of the base case, we conclude that \eqref{eq:many_differences} holds with the desired probability in phase $t$.
	\end{proof}

	\section{Proof of Proposition\texorpdfstring{ \MakeLowercase{\ref{prop:efficient_clean_up}}}{} } 
	\label{appendix:efficient_clean_up}
	
	The goal of this subsection is to prove Proposition \ref{prop:efficient_clean_up}. The next technical lemma implies that a negligible fraction of points are lost in step 4(b), the clean-up step of \textbf{PRDC}.  
	
	\begin{lemma}
		\label{lem:ell_2_argument}
		Let $\alpha = \ceil{s^{1/(4m)}}^{-1}$, and let $\mc{U} = \uu_1, \ldots, \uu_s \stackrel{iid}{\sim} \mr{Tri}[-\alpha, \alpha]^m$ denote a sample from a triangular distribution. Let $v\rp{0} \in \mathbb{R}^m$ denote a random vector independent of $\mc{U}$ satisfying $\abs{v\rp{0}}_2 \leq R m^{3/2}$ for some absolute constant $R > 0$. For $k = 1, 2, \ldots, $ define a sequence of random vectors
		\[
		v\rp{k} =  v\rp{k-1} + a^* \uu_k 
		\]
		where
		\[
		a^* = \argmin_{a \in \{ \pm 1 \} } \abs{ v\rp{k-1} + a \uu_k }_2.
		\]
		Let $c^*$ denote the absolute constant from Claim \ref{claim:submg}. Suppose that $R'\geq 2/c^*$ and
		\[
		K \geq  \frac{8R^2 m^2 \sqrt{s}}{R'c^*}.   
		\]
		Then with probability at least
		\[
		1 - \exp\left( - \frac{(c^*)^2K}{8m} \right)
		\]
		there exists $k \leq K$ such that
		\[
		|v\rp{k}|_2 \leq R' m \alpha.
		\]
	\end{lemma}
	
	\begin{proof}
		By the definition of $v\rp{k}$, we have that
		\[
		0 \leq \abs{v\rp{K+1}}_2^2 = \abs{v\rp{0}}_2^2 + \sum_{k = 0}^K \left(-2\abs{\langle v\rp{k}, \uu_{k+1} \rangle} + \abs{\uu_{k + 1}}_2^2\right). 
		\]
		
		Consider the event $\mc{E}$ that for all $1 \leq k \leq K$, we have $\abs{v\rp{k}}_2 \geq R' m \alpha$. Let $\nu\rp{k} = v\rp{k}/\abs{v\rp{k}}_2$. Observe that $\abs{\uu_k}_2^2 \leq \alpha^2 m$. Applying this and rearranging the inequality above, we have that the event $\mc{E}$ implies
		\begin{equation}
		\label{eq:rare_event_bdd}
		\sum_{k = 0}^K \abs{\langle \nu\rp{k}, \uu_{k + 1} \rangle} \leq \frac{R^2m^3 + \alpha^2  m K}{2R'm\alpha}.
		\end{equation}
		
		For $0 \leq j \leq K$, define a sequence of random variables
		\[
		M_j := \sum_{k = 0}^j \left( \abs{\langle \nu\rp{k}, \uu_{k+1} \rangle } - c^* \alpha \right).
		\] 
		For convenience, we also define $M_{-1} \equiv 0$. Note that $M_j$ is measurable with respect to the sigma-field $\Omega_j$ generated by the random variables $v\rp{0}, v\rp{1}, \ldots, v\rp{j+1}$. Therefore, $\Omega_{-1} \subset \Omega_0 \subset \ldots $ defines a filtration for the sequence of random variables $\{ M_j \}_{j \geq -1}$. 
		\begin{claim}
			\label{claim:submg}
			There exists an absolute constant $c^* > 0$ such that $\{M_j\}_{j \geq -1}$ is a submartingale with respect to the filtration $\{ \Omega_j\}_{j \geq -1}$.
		\end{claim}
		
		\begin{proof}
			Since $v\rp{0}$ is independent of $\mc{U}$ and $\mc{U}$ is an independent sample, it follows that $\uu_{k+1}$ is independent of $\nu\rp{k}$. 
			Observe that the coordinates of $\uu_{k+1}$ are subGaussian. By the Khintchine inequality for the $\ell_1$ norm \citep[see Exercises 2.6.5 and 2.6.6 of][]{Ver18}, we have
			\[
			\mb{E}\left[ \abs{\langle \nu\rp{k}, \uu_{k+1} \rangle } \, \, \bigg| \, \, v\rp{k} \right] = \mb{E}\left[ \abs{\langle \nu\rp{k}, \uu_{k+1} \rangle } \, \, \bigg| \, \, \nu\rp{k} \right] \geq \alpha c^* \abs{\nu\rp{k}}_2 = \alpha c^* > 0
			\]
			for an absolute constant $c^* > 0$. 
		\end{proof} 
		Let $c^* >0$ denote the absolute constant from Claim \ref{claim:submg}, and set $R' \geq 2/c^*$. Next, note the equivalence between the following inequalities:
		\begin{align}
		c^* \alpha K &\geq \frac{c^* \alpha K}{2} + \frac{R^2 m^3 + \alpha^2 m K}{2 R' m \alpha} \Leftrightarrow \label{eqn:Azuma_condition} \\
		K &\geq \frac{R^2 m^2}{R'(c^* - 1/R')} \alpha^{-2}, \nonumber
		\end{align}
		assuming that $c^* - 1/R' > 0$. Setting $R' \geq 2/c^*$, it follows that if
		\[
		K \geq \frac{8 R^2m^2 \sqrt{s} }{R'c^*},
		\]
		then \eqref{eqn:Azuma_condition} holds. Next, note by Cauchy-Schwarz that the submartingale $M_j$ has increments bounded by $\alpha \sqrt{m}$. Since \eqref{eqn:Azuma_condition} holds, we may apply the Hoeffding--Azuma inequality to conclude that for such choice of $K$ and $R'$ that
		\[
		\p[\mc{E}] \leq \p\left[ M_K  \leq \frac{R^2 + \alpha^2  m^2 K}{2R'm\alpha} - c^*\alpha K \right] \leq \p\left[ M_K \leq - \frac{c^* \alpha K}{2} \right] \leq
		\exp\left( - \frac{(c^*)^2 K}{8m} \right),
		\]
		as desired.
			
%
%
%
		
	\end{proof}
	
	
	\begin{proof}[Proof of Proposition \ref{prop:efficient_clean_up}]
			
		Let $t \geq 1$ denote the current phase. Let $\mc{E}$ denote the event that $|S_j| = n_j$ for all $1 \leq j \leq t$ and $|G'_t| = g'_t$ where $n_j \geq (0.3)^{j-1} n$ for all $1 \leq j \leq t$ and $g'_t \geq (0.4) n_t$. By Proposition \ref{prop:GKK_sizes_condl_distr} and Lemma \ref{lem:clean_up_distr} in Appendix \ref{app:distribution_properties}, conditionally on $\mc{E}$, the points $\z_1, \ldots, \z_{g'_t} \in G'_t$ are distributed as $\mr{Tri}[-\alpha_{t+1}, \alpha_{t+1}]^m$, and the leftover vector $v_t\rp{0}$ obtained in step 4(a) of \textbf{PRDC} is independent of this sample. Moreover, by Lemma \ref{lem:reduce} and the fact that $\abs{v_t}_\infty \leq \abs{v_t}_2 \leq \gamma m \alpha_t$, it follows that \[\abs{v_t\rp{0}}_\infty \leq (\gamma + 1) m \alpha_t.\] Hence, the Cauchy--Schwarz inequality yields that \[\abs{v_t\rp{0}}_2 \leq (\gamma + 1) m^{3/2} \alpha_t.\] 
		
		Next, apply Lemma \ref{lem:ell_2_argument} with $\mc{U} = \frac{1}{\alpha_t} \z_1, \ldots, \frac{1}{\alpha_t} \z_{g'_t}$, $v\rp{0} = \frac{1}{\alpha_t} v_t\rp{0}$, $R = \gamma + 1$, $R' = \gamma$, and $K = (g'_t)^{3/4}$ where $\gamma \geq 2/c^*$. Recall that by assumption $g'_t \geq (0.4)n_t \geq (0.4)(0.3)^{t-1} n$. Since $t \leq \ceil{C^* \log n}$, we have that $g'_t \geq \sqrt{n}$. So for $C$ sufficiently small in the bound $m \leq C \sqrt{\log{n}}$, we have that the lower bound		
		\[
		K =(g'_t)^{3/4} \geq \frac{8 (\gamma + 1)^2 m^2 \sqrt{g'_t}}{\gamma c^*}
		\]
		holds, and so indeed Lemma \ref{lem:ell_2_argument} applies. Therefore, conditioned on $\mc{E}$, with probability at least
		\[
		1 - \exp\left( -\frac{(c^*)^2 (g'_t)^{3/4}}{8m} \right) \geq 1 - \exp\left( -(c^*)^2 n^{1/4} \right) 
		\]
		there exists $k \leq K = (g'_t)^{3/4}$ with 
		\[
		\abs{v_t\rp{k}}_2 \leq \gamma m \alpha_{t+1}. 
		\]
		By the lower bounds $n \geq e^{(1/C)m^2}$ and $g'_t \geq \sqrt{n}$, for $C$ sufficiently small, it follows that $(g'_t)^{3/4} \leq (0.01) g'_t$. Hence, conditioned on $\mc{E}$, with probability at least $1 - \exp\left( -(c^*)^2 n^{1/4} \right)$ we have $|S_{t+1}| \geq g'_t - (g'_t)^{3/4} \geq (0.3) n_t$, as desired. 
	\end{proof}
	
	
	
	%
	
	\section{Proof of Theorem\texorpdfstring{  \MakeLowercase{\ref{thm:gkk}} }{} }
	\label{appendix:GKK_thm}
	Our main theorem is a direct consequence of Propositions \ref{prop:many_differences} and \ref{prop:efficient_clean_up}. 
	
	\begin{proof}[Proof of Theorem \ref{thm:gkk}]
		Recall that $T = \ceil{C^* \log{n}}$ where $C^* = (2 \log(10/3) )^{-1},$ and set $\theta = 0.3$. By the union bound over the $T$ phases of \textbf{PRDC} in \textbf{GKK}, induction, and Propositions \ref{prop:many_differences} and \ref{prop:efficient_clean_up}, we have that $|S_t| \geq \theta^{t-1} n$ for all $1 \leq t \leq T$ with probability at least $1 - \exp(-c_3 n^{1/4} )$, for some absolute constant $c_3> 0$. Since $\alpha_{t+1} = \alpha_t/\ceil{|S_t|^{1/(4m)}}$, this implies by induction that 
		\[\alpha_{T} \leq \max(1, \Delta) \theta^{-T^2/(4m)} n^{-T/(4m)} 
		\leq 
		\max(1, \Delta) \exp\left( -\frac{C^* \log^2 n}{8 m} \right)
		\]
		with probability at least $1 - \exp(-c_3 n^{1/4} )$.   
		
		Moreover, by the stopping criterion from step 4(b) of \textbf{PRDC}, $\abs{v_T}_\infty \leq \abs{v_T}_2 \leq \gamma m \alpha_T.$ Applying \textbf{REDUCE} to $S_T \cup \{v_T\}$, we see by Lemma \ref{lem:reduce} that the output $\abs{v}_\infty$ of \textbf{GKK} satisfies
		\[
		\abs{v}_\infty \leq \max(1, \Delta) (\gamma m + m - 1) \exp\left( -\frac{C^* \log^2 n}{8 m} \right) \leq  \exp\left( - \frac{c \log^2{n}}{m} \right) 
		\]
		for an absolute constant $c > 0$. Note that the right-hand-side follows if we take $C > 0$ sufficiently small in the bound $m \leq C \sqrt{\log(n)/\max(1, \log \Delta)}$. 
	\end{proof}
	
	\section{Distributional properties}
	\label{app:distribution_properties}
	Our analysis of \textbf{GKK} relies heavily on the fact that the operations in the algorithm preserve important features of the original distribution such as independence. Though not carefully proven in \cite{KarKar82}, these features are crucial to our analysis, so we provide explicit justification of these properties below for completeness. 
	
	%
	%
	
	First we introduce some notation. Given $\alpha > 0,$ a fixed collection of vectors $\z_1, \ldots, \z_s \subset [-\alpha, \alpha]^m,$ and a density $g:[-\alpha, \alpha]^m$, divide the cube $[-\alpha, \alpha]^m$ into $N := 2^m(\ceil{s^{1/(4m)}})^m$ sub-cubes $C_1, \ldots, C_N$ of side length $\alpha/\ceil{s^{1/(4m)}}$ as in step 1 of \textbf{PRDC}. Label the points $\z_1, \ldots, \z_s$ as in step (2) of \textbf{PRDC} using the density $g$. Define a random collection of ordered pairs $\mc{T}_{s,\alpha,g}  \subset ([N] \times \{0,1\})^s$ so that for $1 \leq i \leq s$, 
	\[
	(\mc{T}_{s,\alpha,g})_i = ( j, 1 )
	\]
	if and only if $\z_i \in C_j$ and if $\z_i$ is labeled `good', and
	\[
	(\mc{T}_{s,\alpha,g})_i = ( j, 0 )
	\]
	if and only if $\z_i \in C_j$ and $\z_i$ is labeled as `bad'.
	
	Usually $s, \alpha$ and $g$ are clear from context, in which case we write $\mc{T}$ for $\mc{T}_{s,\alpha,g}$. Observe that $\mc{T}$ keeps track of which sub-cube $v_i$ lands in and also whether it was labeled good or bad. We refer to $\mc{T}$ as the \textit{configuration vector} corresponding to the input of \textbf{PRDC}.
	
	We proceed by proving some preliminary lemmas, the first of which states roughly that given random vectors $\z_1, \ldots, \z_s$ with a nice conditional distribution, the good points in each sub-cube $C_j$ have a uniform distribution. 
	
	\begin{lemma}
		\label{lem:uniform_distr}
		Suppose that conditioned on an event $\mc{F}$,
		\begin{itemize}
			\item the random vectors $S = \z_1, \ldots, \z_s \in \mb{R}^m$ are iid, and each vector has the conditional joint density $g:[-\Delta, \Delta]^m \to \mb{R}$.
			\item $S \cup \{v\}$ is a collection of independent random vectors.
		\end{itemize}  Run the first  two steps of \textbf{\emph{PRDC}} with input $S = \z_1, \ldots, \z_s$,$v$, $\alpha = \Delta$, and density $g$. Let $G$ denote the good points, and let $B$ denote the bad points. Then conditioned on $\mc{T}_{s, \Delta, g}$ and $\mc{F}$,
		\begin{itemize}
			\item the random vectors in $B \cup G$ are mutually independent.
			\item For $1 \leq j \leq N$, a given good point in $C_j$ has a uniform distribution on $C_j$. 
		\end{itemize} 
	\end{lemma}

	\begin{proof}
		The first statement follows because (1) $G \cup B = \z_1, \ldots, \z_s$ is an independent sample, conditioned on $\mc{F}$, and (2) the ordered pair $(\mc{T}_{s, \Delta, g})_i$ is generated independently for each $i \in [s]$. Thus it suffices to show, by symmetry and passing to conditional densities, that
		\[
		g(z| \z_1 \in C_j, \z_1 \, \, \mr{good}) = \frac{1}{\mr{Vol}(C_j)} 
		\]
		for all $z \in C_j$.
		By Bayes' rule,
		\begin{align*}
		g(z| \z_1 \in C_j, \z_1 \, \, \mr{good}) &= \frac{\p[ \z_1 \,\, \mr{good} | \z_1 = z, \z_1 \in C_j, \, \mc{F}] \, g(z| \z_1 \in C_j)}{\p[\z_1 \,\, \mr{good} |\z_1 \in C_j, \, \mc{F}]} \\ 
		&= \left(\frac{\min_{x \in C_j} g(x) }{g(z)} \cdot \frac{g(z)}{ \p[\z_1 \in C_j \, | \, \mc{F}] }\right)\bigg/ \left(\frac{\mr{Vol}(C_j) \min_{x \in C_j} g(x)}{\p[\z_1 \in C_j| \, \mc{F}]}\right) \\
		&= \frac{1}{\mr{Vol}(C_j)}, 
		\end{align*}
		where the last line follows because
		\[
		\p[\z_1 \,\, \mr{good}, \z_1 \in C_j| \, \mc{F}] = \int_{C_j} \p[\z_1 \, \, \mr{good}| \z_1 = z, \, \mc{F}] g(z) \, dz = \mr{Vol}(C_j) \min_{x \in C_j} g(x). 
		\] 
	\end{proof}
	
	\begin{lemma}
		\label{lem:tri_distr}
		Consider the set-up of Lemma \ref{lem:uniform_distr}, and let $\alpha' = \alpha/\ceil{s^{1/(4m)}}$. Let $G'$ denote the set of random differences constructed after step 3. of \textbf{\emph{PRDC}} applied to $S$, $v$, $\alpha = \Delta$, and $g$. Then conditioned on the events $\mc{F}$ and $\mc{T} = \mf{T}$, the points in $G'$ are iid and have a triangular distribution on $[-\alpha', \alpha']^m$.
	\end{lemma}
	
	\begin{proof}
		Observe that $\mf{T}$ determines the number of points in $G'$. The points in $G'$ are independent by Lemma \ref{lem:uniform_distr} and the fact that the points in $G$ are randomly differenced in step 3. of \textbf{PRDC}. Since $C_j$ is a translation of the sub-cube $[-\alpha', \alpha']^m$, the difference of two independent, uniformly sampled points from $C_j$ have a triangular distribution on $[-\alpha', \alpha']^m$.
	\end{proof}
	
	\begin{lemma}
		\label{lem:points_lost_distr}
		Consider the set-up of Lemma \ref{lem:tri_distr}, and let $\ell \in \mb{Z}_{\geq 0}$. Let the random variable $\mc{L}$ denote the number of points removed from $G'$ in step 4(b) of \textbf{\emph{PRDC}} applied to $S$, $v$, $\alpha = \Delta$, and $g$. Let $S'$ and $v'$ denote the vectors output by \textbf{\emph{PRDC}}. Let $g' = |G'|$. Then conditioned on the events $\mc{F}$, $\mc{T} = \mf{T}$, and $\mc{L} = \ell$,
		\begin{itemize}
			\item The $g'- \ell$ points in $S'$ are iid and follow a triangular distribution on $[-\alpha', \alpha']^m$.
			\item The random vector $v'$ is independent of the vectors in $S'$.
		\end{itemize}   
	\end{lemma}

	\begin{proof} Recall that $|G'| = g'$ is determined by $\mf{T}$. Label the points in $G'$ independently at random to be $G' = \y_1, \ldots, \y_{g'}$. The points in $G'$ are independent and triangularly distributed on $[-\alpha',\alpha']^m$ by Lemma \ref{lem:tri_distr}, conditionally on $\mc{F}$ and $\mc{T} = \mf{T}$. Recall the single vector $v$ that was input initially to \textbf{PRDC}. In step 4(a), this is combined with vectors in $B'$ to construct a single vector $v\rp{0}$. By Lemma \ref{lem:uniform_distr}, we have that $v\rp{0}$ is independent of $G'$, conditionally on $\mc{T} = \mf{T}$ and $\mc{F}$.
		
		Now in step 4(b) of \textbf{PRDC}, let us remove points from $G'$ in the order $\y_{g'}, \y_{g' - 1}, \ldots, \y_{g' - \ell + 1}$. By the stopping criterion for step 4(b), we have
		\[
		\{ \mc{L} = \ell \} = \left\{ \abs{v\rp{k}}_2 > \gamma m \alpha' \, \, \forall \, 1 \leq k \leq \ell-1, \, \, \abs{v\rp{\ell}}_2 < \gamma m \alpha'  \right\}.
		\]
		
		Since $v\rp{k} = v\rp{k-1} \pm \y_{g' - k + 1}$ for $1 \leq k \leq \ell$, the random vector $v\rp{k}$ is independent of $\y_1, \ldots, \y_{g' - \ell}$. Therefore, the sample $S' = \y_1, \ldots, \y_{g' - \ell}$ is independent of the event $\mc{L} = \ell$. Hence, further conditioning on $\mc{L} = \ell$ does not affect the distribution of $S'$, as desired.
	\end{proof}
	Summarizing the content of Lemmas \ref{lem:uniform_distr}, \ref{lem:tri_distr}, and \ref{lem:points_lost_distr}, we have the following proposition.
	
	\begin{proposition}
		\label{prop:PRDC_config_vector_condl}
		Suppose that conditioned on an event $\mc{F}$,
		\begin{itemize}
			\item the random vectors $S = \z_1, \ldots, \z_s \in \mb{R}^m$ are iid, and each vector has the conditional joint density $g:[-\Delta,\Delta]^m \to \mb{R}$.
			\item $S \cup \{v\}$ is a collection of independent random vectors.
		\end{itemize} 
		Let $S',  v'$ denote the vectors output by \textbf{\emph{PRDC}} applied to $S,$ $v,$ $\alpha = \Delta$, and $g$. Let $s' \in \mb{Z}_{\geq 0}$ and $\alpha' = \alpha/\ceil{s^{1/(4m)}}$. Then conditioned on $\mc{F}, \mc{T} = \mf{T}$, and $|S'| = s'$, 
		\begin{itemize}
			\item the $s'$ points in $S'$ are iid and follow a triangular distribution on $[-\alpha', \alpha']^m$.
			\item The random vector $v'$ is independent of the vectors in $S'$.
		\end{itemize}   
	\end{proposition}
	
	Observe that Proposition \ref{prop:PRDC_config_vector_condl} and induction imply the next lemma, which guarantees that we have a nice distribution after every phase of \textbf{PRDC}, conditionally on the data $\mc{T}\rp{j}$ at each step.
	
	\begin{lemma}
		\label{lem:GKK_config_vector_condl}
		Let $X_1, \ldots, X_n$ be iid random vectors, each having a joint density $g:[-\Delta,\Delta]^m \to \mb{R}$, conditioned on some event $\mc{F}$. Consider the output $S_t, v_t, \alpha_t$ that results after the $(t-1)$-th phase of \textbf{\emph{PRDC}} in step 2 of \textbf{\emph{GKK}}. For $1 \leq j \leq t -1$, let $\mc{T}\rp{j}$ denote the configuration vector resulting from step 2 of the $j$-th phase of \textbf{\emph{PRDC}}. Then conditioned on $\mc{T}\rp{j} = \mf{T}\rp{j}$ for $1 \leq j \leq t-1$ and $|S_j| = n_j$ for $1 \leq j \leq t$, we have
		\begin{itemize}
			\item the $n_t$ points in $S_t$ are iid and follow a triangular distribution on $[-\alpha_t, \alpha_t]^m$.
			\item The random vector $v_t$ is independent of the vectors in $S_t$.
		\end{itemize} 
	\end{lemma}
	
	Next, marginalizing over all possible configuration vectors yields Proposition \ref{prop:GKK_sizes_condl_distr}.
	
	
	
	\begin{proof}[Proof of Proposition \ref{prop:GKK_sizes_condl_distr}]
		We induct on the phase $t$. Consider the base case $t = 2$. Let $\z_1, \ldots, \z_{n_2}$ denote the vectors in $S_2$, and let $I_i$ denote a measurable subset of $[-\alpha_2, \alpha_2]^m$ for $1 \leq i \leq n_2$. Recall that $\mf{T}\rp{1}$ determines the number of differences in $G_1'$, and $|S_2|$ determines the amount of points lost in step 4(b) of \textbf{PRDC}. Then we have, marginalizing over all possible choices of $\mf{T}\rp{1}$ compatible with $|S_2| = n_2$,  
		\begin{multline*}
		\p\left[\z_i \in I_i \, \forall \, 1 \leq i \leq n_2 \bigg| \, \, |S_2| = n_2 \right] \\= 
		\sum_{ \mf{T}\rp{1} } \p\left[\z_i \in I_i \, \forall \, 1 \leq i \leq n_2 \bigg| \mc{T}\rp{1} = \mf{T}\rp{1}, \, \, |S_2| = n_2\right] 
		\p\left[ \mc{T}\rp{1} = \mf{T}\rp{1} \bigg|  \, \, |S_2| = n_2\right]
		\end{multline*}
		By Lemma \ref{lem:GKK_config_vector_condl}, 
		\begin{equation*}
		\p\left[\z_i \in I_i \, \forall \, 1 \leq i \leq n_2 \bigg| \mc{T}\rp{1} = \mf{T}\rp{1}, \, \, |S_2| = n_2\right] = 
		\p\left[\uu_i \in I_i \, \forall \, 1 \leq i \leq n_2 \right]
		\end{equation*}
		where $\uu_1, \ldots, \uu_{n_2} \stackrel{iid}{\sim} \mr{Tri}[-\alpha_2, \alpha_2]^m$. Hence,
		\begin{equation*}
		\p\left[\z_i \in I_i \, \forall \, 1 \leq i \leq n_2 \bigg| \, \, |S_2| = n_2\right] = 
		\p[\uu_i \in I_i \, \forall \, 1 \leq i \leq n_2 ],
		\end{equation*}
		which confirms the first bullet point of Proposition \ref{prop:GKK_sizes_condl_distr} for the base case $t = 2$. Following a similar marginalization procedure, this also implies by Lemma \ref{lem:GKK_config_vector_condl} that $v_2$, the single vector output by \textbf{PRDC}, is independent of $S_2$ conditionally on $|S_2|$. 
		
		Now we handle the inductive step. Let $S_t = \y_1, \ldots, \y_{n_t}$ and $v_t$ denote the vectors output by the $(t-1)^{\mathrm{th}}$ phase of \textbf{PRDC}. Suppose that conditionally on $\mc{F} := \{|S_2| = n_2, \ldots, |S_t| = n_t\}$ that $S_t$ is an iid sample of triangularly distributed vectors on $[-\alpha_t, \alpha_t]^m$, and $v_t$ is independent of $S_t$. By Proposition \ref{prop:PRDC_config_vector_condl}, conditionally on $\mc{F}$, $|S_{t+1}| = n_{t+1}$, and the configuration vector $\mc{T}\rp{t} = \mf{T}\rp{t}$, the sample $S_{t+1}$ is an iid collection of triangularly distributed vectors on $[-\alpha_{t+1}, \alpha_{t+1}]^m$. Hence, conditioning on $\mc{F} \cup \{ |S_{t+1}| = n_{t+1} \}$ and applying the same marginalization over the configuration vector $\mf{T}\rp{t}$ as in the base case yields the first bullet point of Proposition \ref{prop:GKK_sizes_condl_distr} for the inductive step. The second bullet point follows similarly.
	\end{proof}
	
	The next lemma is used in Appendix \ref{appendix:efficient_clean_up}. We omit its proof because it is similar to that of Proposition \ref{prop:GKK_sizes_condl_distr}.
	
	
	\begin{lemma}
		\label{lem:clean_up_distr}
		Let $X_1, \ldots, X_n$ be iid random vectors, each having a joint density $g:[-\Delta, \Delta]^m \to \mb{R}$. Apply \emph{\textbf{GKK}} to the matrix $\mf{X}$ with columns $X_1, \ldots, X_n$, and consider the good points $G_t'$ created from random differencing in step 3 of the $t^{\mr{th}}$ phase of \emph{\textbf{PRDC}}. Also consider the random vector $v_t\rp{0}$ formed in step 4(a) of \emph{\textbf{PRDC}}. Then conditioned on $|S_j| = n_j$ for $1 \leq j \leq t$ and $|G_t'| = g'_t$ , 
		\begin{itemize}
			\item the random vectors in $G_t'$ form an independent sample of size $g'_t$ from $\mr{Tri}[-\alpha_{t+1}, \alpha_{t+1}]^m$.
			\item The random vector $v\rp{0}_t$ is independent of the vectors in $G'_t$.
		\end{itemize} 
	\end{lemma}

\medskip 

\noindent \textbf{Acknowledgments} We thank Christopher Harshaw, Abba Krieger, and Tselil Schramm for useful conversations on discrepancy and anonymous reviewers for their helpful suggestions.

\bibliographystyle{plainnat}
\bibliography{bib}

\end{document}